\documentclass{amsart}
\usepackage{graphicx}
\usepackage{amscd}
\usepackage{amsmath}
\usepackage{amsfonts}
\usepackage{amssymb}
\usepackage{bbm}
\usepackage{setspace}
\usepackage{enumerate}         % better lists\
\usepackage{fixme}
\usepackage{color}
\usepackage{url}
\usepackage{amsthm}
\usepackage{bm}
\usepackage{xy}
\usepackage{enumitem}
\usepackage{mathrsfs}
\usepackage[gen]{eurosym}

\theoremstyle{plain}
\newtheorem{theorem}{Theorem}[section]

\newtheorem{corollary}[theorem]{Corollary}

\newtheorem{lemma}[theorem]{Lemma}

\newtheorem{proposition}[theorem]{Proposition}

\newtheorem{definition}[theorem]{Definition}

\newtheorem{assumption}[theorem]{Assumption}
\theoremstyle{remark}
\newtheorem{remark}[theorem]{Remark}

\numberwithin{equation}{section}

\newcommand{\ind}{1\!\kern-1pt \mathrm{I}}
\newcommand{\rsto}{]\!\kern-1.8pt ]}
\newcommand{\lsto}{[\!\kern-1.7pt [}

\numberwithin{equation}{section}

\renewcommand{\rho}{\varrho}

\DeclareMathOperator*{\argmax}{arg\,max}

\begin{document}
\title[Universal portfolios, SPT and the num\'eraire portfolio]{Cover's universal portfolio, stochastic portfolio theory and the num\'eraire portfolio}
\begin{abstract} 

Cover's celebrated theorem states that the long run yield of a properly chosen ``universal'' portfolio is as good as the long run yield of the best \emph{retrospectively chosen} constant rebalanced portfolio. The ``universality'' pertains to the fact that this result is \emph{model-free}, i.e., not dependent on an underlying stochastic process. We extend Cover's theorem to the setting of stochastic portfolio theory as initiated by R.~Fernholz: the rebalancing rule need not to be constant anymore but may depend on the present state of the stock market. This \emph{model-free} result is complemented by a comparison with the log-optimal num\'eraire portfolio when fixing a stochastic model of the stock market. Roughly speaking, under appropriate assumptions, the optimal long run yield coincides for the three approaches mentioned in the title of this paper. We present our results in discrete and continuous time.
\end{abstract}
\thanks{}

\author{Christa Cuchiero, Walter Schachermayer and Ting-Kam Leonard Wong}
\address{Vienna University, Oskar-Morgenstern-Platz 1, A-1090 Vienna}
\address{University of Southern California, Los Angeles, CA 90089, United States}
\date{\today}
\maketitle

\section{Introduction}\label{sec.1}
In \cite{FK:09} the question was raised whether there is a relation between T.~Cover's theory of universal portfolios (which appeared as the very first paper of the present journal, see \cite{C91}) and stochastic portfolio theory as initiated by R.~Fernholz (see \cite{F:02} and earlier references therein). After all, both theories ask for rather general recipes for choosing a good (at least in the long run) \emph{long-only} portfolio among $d$ assets, whose prices over time are given by

\[
S = (S^1_t, \dots, S^d_t).
\]
Here the time $t$ varies in $\mathbb{T}$, where $\mathbb{T}$ stands either for $\mathbb{N} = \{0, 1, \ldots\}$ (discrete time) or $\mathbb{R}_+ = [0, \infty)$ (continuous time). In many cases $S$ is modeled by a stochastic process defined on some probability space. We note, however, that one may also consider a \emph{model-free} approach, where $S=(s^1_t, \dots s^d_t)_{t \in \mathbb{T}}$ is just an arbitrary deterministic sequence in $(0, \infty)^d$. Indeed, Cover's results make sense in a {\it sure} way (as opposed to an {\it almost sure way}) and do not necessarily refer to a probability measure. In a similar spirit model-free results in the context of stochastic portfolio theory were obtained by S.~Pal and L.~Wong \cite{PW:15} in discrete time and A.~Schied et al.~\cite{SSV:16} in continuous time.

The relationship between stochastic portfolio theory and Cover's theory of universal portfolios has recently started to be pursued by T.~Ichiba and M.~Brod \cite{IB:14, B:14} as well as L.~Wong \cite{wong2015universal}. These works shed light on the connection of universal portfolios with the so-called \emph{functionally generated portfolios} in the sense of \cite{F:02} from different angles. In particular, Wong extends Cover's approach to the family of functionally generated portfolios in discrete time and shows that the distribution of wealth in this family satisfies a pathwise large deviation principle as time tends to infinity.

\subsection{Summary of the main results}

In this paper we also consider the setting of stochastic portfolio theory, in the sense that we are interested in the performance of \emph{relative} wealth processes, where \emph{relative} means here ``in comparison with the \emph{market portfolio}''. In other words, we choose the market portfolio to be the num\'eraire, so that the primary assets are the market weights which are supposed to take values in the open $d$-simplex defined by $\Delta^d=\{x \in (0,1)^d \, |\, \sum_{i=1}^d x^i =1\}$. Its 
closure is denoted by $\bar{\Delta}^d= \{x \in [0,1]^d \, |\, \sum_{i=1}^d x^i =1\}$.

Let us start by summarizing first our results in \emph{discrete} time. In contrast to \cite{wong2015universal}, we extend Cover's universal portfolio to a class of \emph{$M$-Lipschitz functions} 
denoted by $\mathcal{L}^M$. Each element of $\mathcal{L}^M$ maps the markets weights to \emph{long-only} portfolio weights in $\bar{\Delta}^d$, 
whose values specify the proportions of current wealth invested in each of the asset (see Definition \ref{def:LM}). 
This enables us to go beyond \emph{functionally generated portfolios} since the portfolio maps are no longer restricted to (normalized) gradient maps of concave functions on the simplex (see \cite{PW:15} for financial and mathematical justifications of this seemingly restrictive condition).

Denoting by $(V^{\pi}_t)_{t =0}^{\infty}$ the \emph{relative wealth process} corresponding to a portfolio strategy\footnote{Here, the portfolio weight $\pi_t$ is chosen at time $t - 1$ and is used over the time interval $[t - 1, t]$.} $(\pi_t)_{t =1}^{\infty}$, we are interested in comparing the asymptotic growth rates, i.e., 
\[
\lim_{T\to \infty} \frac{1}{T} \log( V^{\pi}_T),
\] 
for certain 	``optimal'' portfolio choices $\pi$. 
More precisely, assuming the market weights evolve according to an ergodic time-homogenous Markov process,  
we establish asymptotic equality of the following growth rates:
\begin{itemize}
\item the \emph{best retrospectively chosen portfolio} at time $T$ in the class $
\bigcup^\infty_{M=1}\mathcal{L}^M$, (in this context $V_T^{\ast,M}$ will denote the relative wealth at time $T$ achieved by investing according to the best strategy in $\mathcal{L}^M$ over the time interval $[0, T]$);
\item the \emph{analog of Cover's universal portfolio} whose relative wealth process is denoted by $(V_t(\nu))_{t = 0}^{\infty}$ (here $\nu$ is a probability measure on $\bigcup^\infty_{M=1}\mathcal{L}^M$ with full support on each $\mathcal{L}^M$); 
\item the \emph{log-optimal num\'eraire portfolio} among the class of \emph{long-only} strategies, whose relative wealth process is denoted by  
$(\hat{V}_t)_{t =0 }^{\infty}$.
\end{itemize}

In mathematical terms, our main result thus essentially -- omitting the precise setting and technical conditions -- reads as follows: 

\begin{theorem}\label{th:mainess}
Let $(\mu_t)_{t =0}^{\infty}$ be a time-homogenous ergodic Markov process describing 
the dynamics of the market weights. Then
\begin{align}\label{eq:1}
\lim_{M \to \infty} \lim_{T \to \infty} \frac{1}{T} \log(V_T^{\ast,M})=\lim_{T\to \infty}\frac{1}{T} \log(V_T(\nu))=
\lim_{T\to \infty}\frac{1}{T}\log( \hat{V}_T)
\end{align}
holds almost surely. 
\end{theorem}

One new and appealing feature of the above theorem is the connection with the (long-only) log-optimal portfolio which could not be established in the setups considered so far, since the class of admitted portfolio strategies to build the universal portfolio was simply too small. Indeed, in Cover's original setting where only constant rebalanced portfolios are considered, 
an analogue of Theorem \ref{th:mainess} is true only when the stock returns are i.i.d. 
As the definition of the log-optimal portfolio requires necessarily a probabilistic setting, 
the above theorem is formulated for certain stochastic models in an almost sure sense. 
The first equality in \eqref{eq:1} is however proved to hold in a pathwise sense for \emph{all} trajectories of the market weights (see Theorem \ref{t4.6}).

In \emph{continuous} time, such a pathwise approach is in general no longer possible due to the appearance of stochastic integrals. 
Therefore, we come back to the setting of \emph{functionally generated portfolios} precisely specified in Section \ref{sec:functionally}, 
under which portfolio wealth processes can be defined via H.~F\"ollmer's pathwise approach to stochastic integration in a probability free way 
(see \cite{F:81} and also the recent paper \cite{SSV:16}).  The choice of functionally generated portfolios is in this continuous time setting best possible 
for the following two reasons: first it perfectly connects Cover's theory
with stochastic portfolio theory in continuous time, and second functionally generated portfolios are in the present Markovian context also the largest class for
which wealth processes can be defined in a pathwise way without passing to rough paths theory.
By replacing the set $\mathcal{L}^M$ by certain spaces of functionally generated portfolios (defined in Section \ref{sec:functionally}) 
and assuming that the log-optimal portfolio is functionally generated (for which we provide precise conditions in Proposition \ref{prop:logoptfunctionally}
when the market weights follow an It\^o-diffusion),  we get essentially the same theorem as above, with the first equality holding in a pathwise sense.

Let us remark that in continuous time generalizations of Cover's universal portfolio to non-parametric families of portfolio maps have not been considered so far. 
It thus significantly extends the results obtained by F.~Jamshidian \cite{J:92}, whose setting corresponds to constant portfolio maps in continuous time. 

While our approach is more theoretical, universal portfolio strategies have been studied extensively in an algorithmic framework. See \cite{LH:14} for a recent survey and in particular \cite{HK:15}.

The remainder of the paper is organized as follows. 
In Section \ref{sec:2} we provide a brief overview (in discrete time for convenience) of the main topics of this paper, 
i.e., we review Cover's theorem, introduce the setting of stochastic portfolio theory and recall the notion of the log-optimal num\'eraire portfolio. 
In Section \ref{sec:discretetime} we establish Theorem \ref{th:mainess} in discrete time (see Theorem \ref{tG3} and Corollary \ref{c4.8}), 
while Section \ref{sec:functionally} is dedicated to prove the corresponding statement in continuous time in the setting of functionally generated portfolios
and under the assumption that the market weights follow an ergodic It\^o-diffusion (see Theorem \ref{th:continouscase1} and Corollary \ref{cor:main}). 

\section{Cover's universal portfolio, stochastic portfolio theory and the log-optimal portfolio} \label{sec:2}
\subsection{Cover's universal portfolio}
Cover's insight reveals the striking phenomenon that the ``wisdom of hindsight'' does not give any significant advantage as compared to a properly chosen ``universal'' portfolio which is constructed in a predictable way. The relevant optimality criterion here is the asymptotic growth rate of the portfolio. 

Let us sketch this -- at first glance extremely surprising -- result in a particularly easy setting (compare \cite{C91}, \cite{cover1996universal}): let time $t$ vary in $\mathbb{N}$ and fix $T \in \mathbb{N}.$ Think of an investor who at time $T$ looks back which stock she should have bought at time $t=0$ (by investing her initial endowment of $1\euro$  and subsequently holding the stock). There is an obvious solution to this problem: pick a stock $i \in \{1, \dots, d\}$ which maximizes the performance $\frac{S^i_T}{S^i_0}.$ It clearly also maximizes the normalized logarithmic return
\begin{align}\label{eq:C3}
\frac{1}{T} [\log(S^i_T) - \log (S^i_0)] \qquad i=1, \dots, d.
\end{align}

The ``only'' problem with this trading strategy is, of course, that we have to make our choice at time $t=0$ instead of $t=T$ and cannot have a look into the future:~at time $t=0$ we do not have the knowledge which stock will perform best during $[0,T]$.

Here is the remedy (compare e.g., \cite{blum1999universal}): at time $t=0$ simply divide the initial endowment of $1 \euro$ into $d$ portions of $\frac{1}{d} \euro$, invest each portion in each of the stocks and then hold the resulting portfolio. At time $T$ the wealth equals 
\begin{align*}
V_T=\frac{1}{d} \sum^d_{j=1} \frac{S^j_T}{S^j_0} \geq \frac{1}{d} \frac{S^i_T}{S^i_0},
\end{align*}
where again $i$ denotes the stock which performed best during the time interval $[0,T]$.\footnote{By a slight abuse of notation $V$ stands here for the wealth expressed in Euro units as opposed to relative wealth considered in the introduction and also subsequently.} 

Passing again to normalized logarithmic returns we obtain
\begin{align}\label{eq:C4}
\frac{1}{T} \log (V_T) \geq  \frac{1}{T} \Big[\log (S^i_T) - \log(S^i_0) - \log d\Big],
\end{align}
so that the difference between \eqref{eq:C3} and \eqref{eq:C4} can be bounded by $\frac{\log(d)}{T}$, which tends to zero as $T \to \infty.$ Hence the universal portfolio consisting of equally weighing the $d$ stocks at time $t = 0$, and subsequently following a buy-and-hold strategy, asymptotically has the same normalized logarithmic return as the -- only in retrospect known -- best performing stock.

\vskip10pt
In fact, Cover considers a more interesting and challenging setting than the above almost trivial observation. Instead of only considering ``pure'' investments into one of the stocks as benchmark, he considers all {\it constant rebalanced portfolio strategies}: let $b=(b^1, \dots, b^d)$ be a fixed element of the $d$-simplex $\bar{\Delta}^d$, i.e., 
$b^j \geq 0$ and $\sum^d_{j=1} b^j=1.$ The corresponding constant rebalanced portfolio $ (V_t(b))_{t=0}^{\infty}$ starting at $V_0(b)=1$ is inductively defined by holding the investment $b^j V_t(b)$ in stock $S^j$ during the  period $(t, t+1)$, so that $V_0(b)=1$ and 
\begin{align}\label{p2}
\frac{V_{t+1}(b)}{V_t (b)}(s)=\sum^d_{j=1} b^j \frac{s^j_{t+1}}{s^j_t},
\end{align}
for each scenario $s=((s^j_t)^d_{j=1})^\infty_{t=0}$ of strictly positive numbers modeling the evolution of the prices of $d$ stocks $S^1, \ldots, S^d$.

In other words, $(V_t (b))_{t\geq0}$ denotes the wealth process of an investor who always invests the proportion $b^j$ of her current wealth $V_t(b)$ into stock $j$. The ``pure'' investment into stock $j$ then corresponds to the $j$'th unit vector $b=e_j$. Of course, for a given scenario $(s^1_t, \dots, s^d_t)^T_{t=0}$, there may be weights $b$ such that the performance of $(V_t (b))^T_{t=0}$ is strictly better than the performance of the best stock. Here is an easy example taken from \cite{blum1999universal}: consider a market with just two stocks. The price of one stock remains constant and the price of the other stock alternatively halves and doubles. Investing in a single stock will not increase the wealth by more than a factor of two. However, letting $b=(\frac{1}{2}, \frac{1}{2})$ the ratio
$\frac{V_{t+1}(b)}{V_t(b)}$ equals $\frac{3}{4}$ or $\frac{3}{2}$, depending on whether $t$ is even or odd. Hence $V_{2t}(b)=(\frac{9}{8})^t$, for all $t \geq 1$, so that $(V_t(b))_{t \geq 0}$ increases exponentially.

Fix again $T$ and define in a pathwise way the quantity $V^*_T$ by
\begin{align}\label{C7}
V^*_T(s)= \max_{b \in \bar{\Delta}^d} V_T(b)(s),
\end{align}
which is a function depending on the scenario $s=(s^1_t, \dots, s^d_t)^T_{t=0}.$

Again, the idea is that, with hindsight, i.e., knowing the trajectory $(s^1_t, \dots, s^d_t)^T_{t=0},$ one considers the best weight $b \in \bar{\Delta}^d$ which attains the maximum \eqref{C7}. Cover's goal is to construct a universal portfolio which performs as well as the hypothetical portfolio process $(V^*_T)_{T \geq 0},$ asymptotically for $T \to \infty.$ In order to do so,  let $\nu$ be a probability measure on $\bar{\Delta}^d.$ The measure $\nu$ will replace the previously considered uniform distribution $\frac{1}{d} \sum^d_{j=1} \delta_{e_j}$ on the unit vectors of $\bar{\Delta}^d.$ The universal portfolio now consists of investing at time $0$ the portion $d\nu(b)$ of one's wealth into the constant rebalanced portfolio $V(b)$ and subsequently following the constant rebalanced portfolio process $(V_t(b))^T_{t=0}.$ The explicit formula is

\begin{align}\label{C8}
V_t(\nu)=\int_{\bar{\Delta}^d}V_t(b)d\nu(b),
\end{align}
where $V_t(b)$ is defined via \eqref{p2}.

Cover's celebrated result now reads as follows.

\begin{theorem}{(Cover \cite{C91}):}\label{t1.1}
Let $\nu$ be a probability measure on $\bar{\Delta}^d$ with full support. Then 
\begin{align}\label{p3a}
\lim_{T \rightarrow \infty} \frac{1}{T} \log \frac{V_T(\nu)(s)}{V_T^*(s)} = 0,
\end{align}
for {\it all} trajectories $s=(s^1_t, \dots, s^d_t)^\infty_{t=0}$ for which there are constants $0 < c \leq C < \infty$ such that
\begin{align}\label{p3}
c \leq \frac{s^j_{t+1}}{s^j_t} \leq C,\qquad \text{for all} \quad j=1, \dots, d \quad \text{and all} \quad t \in \mathbb{N}.
\end{align}
\end{theorem}

We provide the (easy) proof in the Appendix. This proof relies on the uniform boundedness assumption \eqref{p3} which has to be imposed on {\it all} $t \in \mathbb{N}.$ This assumption is also made in much of the subsequent literature on Cover's universal portfolio.  It was shown by T.~Cover and E.~Ordentlich \cite{cover1996universal} that this assumption can, in fact be dropped, at least when $\nu$ is chosen to be uniformly or Dirichlet $(\frac{1}{2}, \cdots, \frac{1}{2})$ distributed on $\Delta^d$ (see also  A.~Blum and A.~Kalai \cite{blum1999universal} for an elegant proof in case of the uniform distribution). In this context sharper estimates than \eqref{p3a} can be obtained, depending on the dimension $d$. For further literature on universal portfolio theory we refer the reader to the recent survey article \cite{LH:14} and the references therein.

\subsection{Stochastic Portfolio Theory}\label{sec:stochportfolio}

In stochastic portfolio theory one usually considers a class of portfolio processes which is somewhat different from the above setting. First of all, one associates to each vector of market capitalizations $(s^1, \dots, s^d) \in (0, \infty)^d$ the vector of {\it market weights} $(\mu^1, \dots, \mu^d) \in \Delta^d$ by normalizing with the total market capitalization $s^1 + \dots + s^d,$ i.e., 
\begin{align*}
(\mu^1, \dots, \mu^d)=\left(\frac{s^1}{s^1 + \dots + s^d} , \dots, \frac{s^d}{s^1 + \dots + s^d}\right).
\end{align*}
Economically speaking, this amounts to take the {\it market portfolio} $(\sum^d_{j=1}s^j_t)^\infty_{t = 0}$ as {\it num\'eraire} (compare \cite{DS:95}, \cite{FK:10}).

Second, investment strategies are usually encoded via {\it portfolio maps}. That are (Borel) measurable functions
\begin{align}\label{eq:portfolio map}
\pi:\Delta^d \to \bar{\Delta}^d
\end{align}
which associate to the current market capitalization $\mu_t=(\mu^1_t, \dots, \mu^d_t)$ the weights $(\pi(\mu_t)=(\pi^1(\mu_t), \dots, \pi^d(\mu_t))$ according to which an agent distributes current wealth among the $d$ stocks at time $t$. 
We then obtain (here in discrete time for simplicity) the {\it relative wealth process} $(V^{\pi}_t)^\infty_{t=0},$ expressed in units of the market portfolio $(s^1_t + \cdots + s^d_t)^\infty_{t=0},$ by starting at $V_0=1$ and the following recursive relation\footnote{Here it is assumed implicitly that the stocks do not pay dividends. This assumption is common in universal and stochastic portfolio theory and allows us to focus on the main ideas.}
\begin{align}\label{p4}
\frac{V^{\pi}_{t+1}}{V^{\pi}_t}= \sum^d_{j=1} \pi_{t+1}^j \frac{\mu^j_{t+1}}{\mu^j_t},
\end{align}
where $\pi_{t+1} = \pi(\mu_t)$. The {\it constant rebalanced portfolio strategies} considered by Cover correspond to the constant functions $\pi:\Delta^d \to \bar{\Delta}^d$.

The main theme of the present paper is to extend Cover's theory pertaining to \emph{constant rebalanced portfolios} 
to more general deterministic \emph{portfolio maps}. While it is not possible to develop a theory analogous to Cover's if we allow for \emph{all} 
measurable functions $\pi$ in \eqref{eq:portfolio map}, we can do so by restricting to sufficiently regular functions. 
We shall focus on Lipschitz functions $\pi$ whose values are separated from the boundary of $\bar{\Delta}^d$, namely the class $\mathcal{L}^M$ as mentioned in 
Section \ref{sec.1} and introduced in Definition \ref{def:LM} below. 
This turns out to be the right setting to prove an analogue of Theorem \ref{t1.1} pertaining to this more general class of 
strategies as made precise in Theorem \ref{t4.6} in discrete time. 
In this context let us again mention the paper \cite{wong2015universal} which addresses in discrete 
time Cover's universal portfolio and \emph{functionally generated portfolio maps} in stochastic portfolio theory. 
We take up this theme in continuous time and prove a model-free analogue to Theorem \ref{t1.1} for this class of functions (see Theorem \ref{th:continouscase1}).

\subsection{The log-optimal portfolio} \label{sec:logoptimal}
In the previous section we considered a fixed scenario $(\mu_t)_{t=0}^{\infty}$ of market capitalizations. We now pass to a more probabilistic setting. The stock price process $S=(S^1_t, \dots, S^d_t)^\infty_{t=0}$ and the corresponding relative market capitalizations $\mu=(\mu^1_t, \dots, \mu^d_t)^\infty_{t=0}$ are now assumed to be stochastic processes (for the sake of a convenient exposition here in discrete time) defined on a filtered probability space $(\Omega, \mathcal{F}, (\mathcal{F}_t)^\infty_{t=0}, \mathbb{P})$. As we are only interested in the {\it relative} performance of a portfolio in comparison to the market portfolio, we shall only consider the process $\mu$ in the sequel and discard the process $S$.

There is an ample literature on the log optimal portfolio, which in most cases coincides with the growth optimal or the num\'eraire portfolio (see e.g., \cite{becherer2001numeraire}, \cite{KK:07} and the references given there). For fixed horizon $T$, this portfolio is defined as the process $V^{\pi}$ as of \eqref{p4} which maximizes the expected logarithmic growth rate 
\begin{align}\label{C15}
\mathbb{E}[\log (V_T^{\pi})] =\mathbb{E}\left[\sum_{t=0}^{T-1}\log \left(\sum^d_{j=1} \pi^j_{t+1} \frac{\mu^j_{t+1}}{\mu^j_t}\right)\right].
\end{align}
Here we allow for all predictable, admissible trading strategies $(\pi_t)^{T}_{t=1}$, where the portfolio weight $\pi_t$ is used over the time interval $[t-1, t]$. Note that in \eqref{p4} we restricted ourselves to ``Markovian'' strategies $\pi_{t+1}=\pi(\mu_t)$ given by a {\it deterministic function} $\pi:\Delta^d \mapsto \bar{\Delta}^d$ as in \eqref{eq:portfolio map}. It is well understood that under mild assumptions on the process $\mu$ we may assert (in discrete as well as in continuous time) the existence of a unique optimizer $\pi^{\text{num}}$ to \eqref{C15}, see e.g., \cite{becherer2001numeraire}, \cite{KS:99}.

A natural question is the following: how does this num\'eraire portfolio relate to Cover's universal portfolio in the context of stochastic portfolio theory? After all, the optimizer $\pi^{\text{num}}$ to \eqref{C15} is, by definition, the (predictable) portfolio strategy with the optimal growth rate, relative to the horizon $T$ among ${\it all}$ predictable strategies. 

To give reasonable answers to this question we begin by identifying some appropriate assumptions on the underlying process $\mu$. First of all, we assume that $\mu$ is Markovian, which simplifies the situation considerably: indeed, if $\mu$ is Markovian, the $\mathbb{R}^d$-valued predictable trading strategy $(\pi^{\text{num}}_t)^\infty_{t=1}$ defining $V^{\text{num}}$ is of the form $\pi^{\text{num}}_t=\pi_t(\mu_{t-1}),$ for some deterministic functions $\pi_t$ on $\Delta^d$. A priori these functions may depend on $t$. If we also suppose that $\mu$ is homogeneous in time, the functions $\pi_t$ do not depend on $t$, i.e., we deal with a {\it deterministic function} $\pi:\Delta^d \to H^d$, where $H^d$ denotes the hyperplane $H^d=\{x \in \mathbb{R}^d | \sum^d_{j=1} x^j=1\}$. In contrast to the portfolio maps considered in Section \ref{sec:stochportfolio}, these functions define a fully invested portfolio without the additional requirement of being {\it long-only}, allowing also for shorting, i.e., negative investments into some of the stocks. In order to work with the same strategies as in Section \ref{sec:stochportfolio}, it is natural to focus on {\it long-only, fully invested} portfolios. That is we optimize in \eqref{C15} only over processes taking values in $\bar{\Delta}^d$. If $\mu$ is a time-homogenous Markov process, we therefore end up with a {\it deterministic function} $\pi:\Delta^d \mapsto \bar{\Delta}^d$ as in \eqref{eq:portfolio map}, which we denote by $\hat{\pi}$, i.e., $\hat{\pi}$ is the optimizer of
\begin{align}\label{C16}
\sup_{\pi \in \Pi_{\text{long}}}\mathbb{E}[\log (V_T^{\pi})],
\end{align}
where $\Pi_{\text{long}}$ denotes the set of predictable processes taking values in $\bar{\Delta}^d$. We always use $\pi^{\text{num}}$ to denote the unrestricted log-optimal trading strategy, while $\hat{\pi}$ refers to the long-only log-optimal portfolio.

From an economic point of view the Markovian assumption introduced above can be motivated by the long term stability of capital distributions of equity markets (see \cite[Chapter 5]{F:02}). In stochastic portfolio theory, this observation led for instance to stable diffusion models of the market weights (see, for example \cite{BFI:05} for the Atlas model).

\section{A comparison of the three approaches - the discrete time case}\label{sec:discretetime}

We have recapitulated in the previous section three points of view for choosing  a ``good'' or ``optimal'' portfolio with respect to the long term growth rate. Our basic goal is to relate and compare these concepts. Throughout this section we work in discrete time and assume that the market weights are described by a $d$-dimensional path $\mu=(\mu_t)^\infty_{t=0}$ with values in $\Delta^d$. We consider as far as possible a model-free approach, but will introduce a probabilistic setting when necessary.

\subsection{Different types of portfolios}\label{sec:portfoliodiscrete}
We start by defining rigorously the following portfolios:
\begin{enumerate}
\item [(i)] the best retrospectively chosen portfolio,
\item [(ii)] the universal portfolio, and
\item [(iii)] the long-only log-optimal portfolio.
\end{enumerate}

While (ii) and (iii) pertain to {\it predictable} portfolio choices which by definition are not allowed to look into the future, the choice of (i) is highly non-adaptive.

The salient difference between (ii) and (iii) is that (iii) can only be defined for a given {\it stochastic model} of the process $(\mu_t)^\infty_{t=0}$: By definition, the log-optimal portfolio maximizes the expected growth rate of the (relative) wealth process. On the other hand, the choice of the retrospectively optimal portfolio (i) and the universal portfolio (ii) do not require knowledge of the specific stochastic model underlying the process $(\mu_t)^\infty_{t=0}$. 

\subsubsection{The best retrospectively chosen portfolio}
Let us first come back to Cover's theme of choosing retrospectively at time $T$ a strategy which is optimal within a certain class of strategies. As mentioned in Theorem \ref{t1.1}, Cover originally considered the class of {\it constant rebalanced portfolios}. These strategies correspond to the constant functions $\pi:\Delta^d \to \bar{\Delta}^d.$ We want to extend his approach to more general portfolio maps $\pi:\Delta^d \to \bar{\Delta}^d$ and compare their performance with the log-optimal map $\hat{\pi}:\Delta^d \to \bar{\Delta}^d$ identified in \eqref{C16} (under time-homogenous Markovianity of $\mu$). A moment's reflection reveals that it does not make sense to allow to choose at time $T$ -- when we dispose of the wisdom of hindsight -- among {\it all measurable} functions $\pi:\Delta^d \to \bar{\Delta}^d.$ Indeed, knowing $(\mu_t)^T_{t=0}$ (and assuming that the elements $\mu_t$ are different, for $t=0, \dots, T,$ to avoid trivialities) there is no restriction to choose $\pi$ such that $\pi(\mu_t)=e_{j(t)},$ where $j(t)\in\{1, \dots, d\}$ maximizes $\frac{\mu^j_{t+1}}{\mu^j_t}.$ In other words,  we can, for each point in time $t=0, \dots, T-1,$ {\it independently} choose the best performing stock with respect to the subsequent period $[t,t+1].$ This is asking for too much clairvoyance and does not allow for meaningful results: Compare \cite{cover1996universal} and \cite[Section 5]{blum1999universal}.

However, it does make sense (economically as well as mathematically) to restrict to more regular strategies $\pi:\Delta^d \to \bar{\Delta}^d$ which satisfy a certain continuity property -- after all, if two states of the market are similar, the corresponding optimal portfolios should be similar as well. In particular, we shall work with the following set of $M$-Lipschitz portfolio maps.

\begin{definition}\label{def:LM}
For $M>0$ we define by $\mathcal{L}^M$ the set of all $M$-Lipschitz functions $\Delta^d \to \bar{\Delta}^d_{M^{-1}}.$ Here $\bar{\Delta}^d_\epsilon$ denotes the set of 
$x \in \Delta^d$ satisfying $x^j \geq \frac{\epsilon}{d},$ for $j=1, \dots, d.$ The Lipschitz constant $M$ pertains, e.g., to the metric induced by the $\ell_1$-norm $\|\cdot \|_1$ on $\bar{\Delta}^d$.
\end{definition}

\begin{remark}\label{rem:lip}
Note that the set $\mathcal{L}^M$ of $M$-Lipschitz functions $\pi:\Delta^d \to \bar{\Delta}^d_{M^{-1}}$ is a compact metric space with respect to the topology of uniform convergence induced by the norm $\|\pi\|_{\infty}=\sup \{\|\pi(x)\|_1 :x \in \Delta^d \}.$
\end{remark}

\begin{remark} We also remark that, instead of Lipschitz functions, we could just as well consider other compact function spaces, e.g., H\"older spaces eqipped with a proper norm. This is done in the context of functionally generated portfolios in Section \ref{sec:functionally}.
\end{remark}

The retrospectively chosen best performing portfolio among the above Lipschitz maps is defined as follows:

\begin{definition}
For a given scenario $(\mu_t)^T_{t=0} \in (\Delta^d)^{T+1}$ we define
\begin{align}\label{p7}
V^{*,M}_T=\sup_{\pi \in \mathcal{L}^M} V^{\pi}_T = \sup_{\pi \in \mathcal{L}^M} \prod^{T-1}_{t=0} \left( \sum^d_{j=1} \pi^j(\mu_t) \frac{\mu^j_{t+1}}{\mu^j_t} \right).
\end{align}
\end{definition}

By compactness (see Remark \ref{rem:lip}) there exists an optimizer $\pi^{*,M} \in \mathcal{L}^M$ (not unique) such that $V^{*,M}_T=V_T^{\pi^{*,M}}$, thus the $\sup$ above can be replaced by $\max$. Our aim is to find a {\it predictable} process $\pi^M=(\pi^M_t)^\infty_{t=1}$,\footnote{Here the portfolio vector $\pi_t^M$ is chosen at time $t - 1$ and is used over the time interval $[t - 1, t]$.} such that each $\pi^M_t$ is in $\mathcal{L}^M$ and such that the performance of $(V_t^{\pi^M})^\infty_{t=0}$ is asymptotically as good as that of $(V^{*,M}_t)^\infty_{t=0}$. This property should hold true {\it universally}, i.e., for all sequences $(\mu_t)^\infty_{t=0}$ in $\Delta^d.$ 

\subsubsection{The universal portfolio}
It turns out that we can define a universal portfolio in an analogous way as in \eqref{C8}.
Indeed, since $\mathcal{L}^M$ is a compact metric space,  we may find a (Borel) probability measure $\nu$ on $(\mathcal{L}_M, \|\cdot \|_\infty)$ with full support. Similarly as in Theorem \ref{t1.1} we divide the initial wealth of $1 \euro$ by investing $d\nu(\pi) \euro$ into the wealth process $V^\pi$ given by the portfolio map $\pi.$ We thus obtain, as in \eqref{C8},

\begin{align}\label{B1}
V^M_T(\nu)=\int_{\mathcal{L}^M} V^\pi_T d \nu(\pi).
\end{align}

Though not explicitly used in this section, we mention in passing that the portfolio vector at time $t$ is given by the wealth-weighted average
\[
\frac{\int_{{\mathcal{L}}^M} \pi V_t^{\pi} d \nu(\pi)}{\int_{{\mathcal{L}}^M} V_t^{\pi} d\nu(\pi)}.
\]
This can be interpreted as a posterior mean (see \cite{wong2015universal} and the references therein).

\subsubsection{The log-optimal portfolio}
In order to relate the universal portfolio with the (long-only) log-optimal portfolio, we need some assumptions on the underlying process $\mu$. For the reasons explained in Section \ref{sec:logoptimal}, we shall assume that $\mu=(\mu_t)^\infty_{t=0}$ is a time-homogeneous Markov process.  In addition, we suppose that it satisfies the following assumption.

\begin{assumption}\label{5.1}
The process $\mu$ is a \textup{time homogeneous, ergodic Markov process} with a unique invariant measure $\rho$ on the open simplex $\Delta^d$.
\end{assumption}

We denote the transition kernel of the chain by by $(\rho(x,\cdot))_{x \in \Delta^d}$, i.e., 
for all Borel sets $A\subseteq \bar{\Delta}^d$, we have $\mathbb{P}[\mu_{t+1} \in A| {\mathcal{F}}_t] = \rho(\mu_t, A)$. For further notions concerning ergodic Markov processes we refer to \cite{E:15}.

What is the {\it fully invested, long-only, log-optimal} trading strategy $\hat{\pi}$ for the process $\mu?$ Given that $\mu_t=x \in \Delta^d$, we know the conditional law $\rho(x, \cdot)$ of $\mu_{t+1}.$ We therefore choose $\hat{\pi}(x) \in \bar{\Delta}^d$ as a solution to
\begin{equation}\label{E2}
\begin{split}
\\ \hat{\pi}(x)=&\argmax_{p \in \bar{\Delta}^d} \left(\int_{\Delta^d} \log (\langle p, \frac{y}{x}\big\rangle ) \rho (x, dy)\right)
\\ =&\argmax_{p \in \bar{\Delta}^d} \left(\int_{\Delta^d} \log (\sum^d_{j=1} p^j \frac{y^j}{x^j}) \rho (x, dy)\right)
 \end{split}
\end{equation}
and assume that $\hat{\pi}(\cdot)$ can be chosen to be measurable.

For $x \in \Delta^d,$ define the number $L(x)$ as the value of the optimization problem \eqref{E2}, i.e.,
\begin{equation}\label{4.2}
    \begin{split}
L(x)&=\max_{p  \in \bar{\Delta}^d}\left(\int_{\Delta_d} \log (\langle p, \frac{y}{x} \rangle ) \rho(x, dy)\right)
\\&=\int_{\Delta_d} \log (\langle \hat{\pi}(x), \frac{y}{x} \rangle ) \rho(x, dy).
\end{split}
\end{equation}

By considering $\pi(x)=x$ (which corresponds to the market portfolio) we clearly have $L(x) \geq 0$ for each $x \in \Delta^d.$ We obtain the a.s.~relation 
\begin{align*}
L(x)=\mathbb{E}\left[\log\left (\frac{\hat{V}_{t+1}}{\hat{V}_t}\right) \middle| \mu_t=x \right],
\end{align*}
where $\hat{V}=(\hat{V}_t)^\infty_{t=0}$ denotes the long-only log-optimal wealth process $V^{\hat{\pi}}$ defined by the portfolio map $\hat{\pi}$ via \eqref{p4}.

\begin{assumption}\label{5.2}
Using the above notation we assume that
\begin{align}\label{eq:ass4.2}
L:=\int_{\Delta^d} L(x) d\rho (x) < \infty.
\end{align}
\end{assumption}

This assumption rules out the rather trivial (and unrealistic) case that the expected yield of the log-optimal portfolio equals infinity. We then may apply Birkhoff's ergodic theorem for discrete time Markov processes (see \cite[Theorem 2.2, Section 2.1.4]{E:15} to obtain the following basic result.

\begin{theorem}\label{t4.3}
Under Assumptions \ref{5.1} and \ref{5.2}, we have that, for $\rho$-a.e.~starting value $\mu_0 \in \Delta^d$, 
\begin{equation}\label{E4}
\lim_{T \to \infty} \frac{1}{T} \log (\hat{V}_T)=L,
\end{equation}
the limit holding true a.s.~as well as in $L^1$.

More generally, let $\pi:\Delta^d  \to \bar{\Delta}^d$ be any ($\rho$-measurable) portfolio map such that
\begin{equation}\label{E4a}
L^{\pi}:=\int_{\Delta^d}\left(\int_{\Delta^d} \log \left(\left\langle \pi (x), \frac{y}{x}\right\rangle \right) \rho(x,dy)\right) d \rho(x) > - \infty.
\end{equation}
We then have, for $\rho$-a.s.~starting value $\mu_0,$ that 
\begin{equation}\label{E4b}
\lim_{T \to \infty} \frac{1}{T} \log (V^{\pi}_T)=L^{\pi}
\end{equation}
a.s.~as well as in $L^1$.
\end{theorem}
As mentioned above, the function $\hat{\pi}:\Delta^d \to \bar{\Delta}^d$ is $\rho$-measurable. In particular, the integral \eqref{eq:ass4.2} makes sense.  However, in general there is little reason why the function $\hat{\pi}$ should have better regularity properties than being just $\rho$-measurable. On the other hand, we may {\it approximate} $\hat{\pi}$ by more regular functions, in particular by functions in $\mathcal{L}^M$. This will be crucial for comparing the asymptotic growth rates.

The subsequent result is intuitively rather obvious, but the proof turns out to be quite technical. 

\begin{lemma}\label{5.4}
Under Assumptions \ref{5.1} and \ref{5.2}, for any $\epsilon >0$ there exist $M>0$ and an {\it $M$-Lipschitz function} $\pi_{Lip} \in \mathcal{L}^M$ such that
$$L^{\pi_{Lip}} > L - \epsilon,$$
where $L$ and $L^{\pi}$ are given in \eqref{eq:ass4.2} and \eqref{E4a} respectively. In particular, we have $L = \sup_M \sup_{\pi \in {\mathcal{L}}^M} L^{\pi}$.
\end{lemma}

\begin{proof}
Let $\hat{\pi}: \Delta^d \to \bar{\Delta}^d$ be the optimizer of \eqref{E2} and define, for $0 < \epsilon <1,$
$$\pi_\epsilon=(1-\epsilon)\hat{\pi} + \epsilon \left(\frac{1}{d}, \dots, \frac{1}{d}\right).$$
Note that $\pi_\epsilon$ takes values in $\bar{\Delta}^d_\epsilon$ (see Definition \ref{def:LM}). Also note that, for $p \in \bar{\Delta}^d_\epsilon,$ we have
\begin{align}\label{C1}
\langle p, \frac{y}{x}\rangle= \sum^d_{j=1} p^j \frac{y^j}{x^j} \geq \frac{\epsilon}{d},
\end{align}
for $x,y \in \Delta^d$, as at least one of the terms $\frac{y^j}{x^j}$ is greater or equal to one.

The average performance $L^{\pi_\epsilon}$ defined via \eqref{E4a} for the portfolio map $\pi_\epsilon$ is still almost as good as the optimal average performance $L\equiv L^{\hat{\pi}}$:
\begin{equation}\label{C2c}
    \begin{split}
\\L^{\pi_\epsilon}&=\int_{\Delta^d} \left[ \int_{\Delta^d} \log (\langle \pi_\epsilon(x), \frac{y}{x} \rangle) \rho(x,dy) \right] d\rho(x)
\\&\geq\int_{\Delta^d} \left[\int_{\Delta^d} \log ((1-\epsilon) \langle \hat{\pi} (x), \frac{y}{x}\rangle ) \rho(x,dy)\right] d\rho (x)
\\&\geq L+\log (1-\epsilon).
\end{split}
\end{equation}

To approximate $\pi_\epsilon$ by a Lipschitz function $\pi_{Lip}$ taking its values in $\bar{\Delta}^d_\epsilon,$ we need some preparation. By Assumption \ref{5.2} we can find $\delta >0$ such that, for $A \subseteq \Delta^d,$
\begin{equation}\label{C2a}
\int_A \left[\int_{\Delta^d} (\log(\frac{\epsilon}{d})-\log(\langle \pi_\epsilon(x), \frac{y}{x}\rangle)) \rho(x,dy)\right] d\rho (x) > -\epsilon,
\end{equation}
provided that $\rho [A] < \delta.$ In particular, we may find $\eta>0$ such that
\begin{equation}\label{C2b}
\int_{\Delta^d \backslash \bar{\Delta}^d_\eta} \left[\int_{\Delta^d} (\log(\frac{\epsilon}{d})-\log(\langle \pi_\epsilon(x), \frac{y}{x}\rangle)) d\rho(x,y) \right]d\rho (x) > -\epsilon.
\end{equation}
Now we find a Lipschitz function $\pi_{Lip}:\Delta^d \to \bar{\Delta}^d_\epsilon$ such that
\begin{equation}\label{C2}
\| \pi_{Lip}(x)-\pi_{\epsilon}(x) \|_1=\sum^d_{j=1} | \pi_{Lip}(x)^j - \pi_\epsilon(x)^j| < \eta \epsilon^2,
\end{equation}
for all $x \in \Delta^d \backslash A,$ where the exceptional set $A$ satisfies $\rho [A] < \delta$. Indeed, the functions from $\mathbb{R}^d \to \bar{\Delta}^d_\epsilon$
which are continuously differentiable in a neighborhood of $\Delta^d$ are dense with respect to the $L^1(\mathbb{R}^d, \rho; \mathbb{R}^d)$-norm.
Let $M$ be a Lipschitz constant for $\pi_{Lip}$ such that $M^{-1} \leq \epsilon.$

To estimate $L^{\pi_{Lip}} - L^{\pi_\epsilon}$ we argue separately on the sets $\Delta^d \backslash \bar{\Delta}^d_\eta, A \cap \bar{\Delta}^d_\eta$ and $\bar{\Delta}^d_\eta \backslash A.$ To start with the latter set note that, for $x \in \bar{\Delta}^d_\eta$ and $y \in \Delta^d$  we have that the function
$$p \mapsto \langle p, \frac{y}{x}\rangle = \sum^n_{i=1} p^j \frac{y^j}{x^j}, \qquad p \in \bar{\Delta}^d,$$
is Lipschitz on $\bar{\Delta}^d$ with Lipschitz constant bounded by $(\frac{\eta}{d})^{-1}.$ From \eqref{C2} we get
\begin{equation}\label{C3}
    \begin{split}
\\ \int_{\bar{\Delta}^d_\eta \backslash A} [\int_{\Delta^d} (\log (\langle \pi_{Lip} (x), \frac{y}{x} \rangle) -\log (\langle \pi_\epsilon(x), \frac{y}{x} \rangle) d\rho(x,y)] d\rho(x) 
\\ \geq -(\eta \cdot \epsilon^2)(\frac{\eta}{d})^{-1}(\frac{\epsilon}{d})^{-1} \geq -d^2\epsilon.
    \end{split}
\end{equation}
The term $(\frac{\epsilon}{d})^{-1}$ above comes from the fact that $\langle \pi_{Lip}(x), \frac{y}{x}\rangle$ as well as $\langle \pi_{\epsilon}(x), \frac{y}{x}\rangle$ takes values in $[\frac{\epsilon}{d}, \infty[$ and the function $z \mapsto \log(z)$ is Lipschitz on this set with constant $(\frac{\epsilon}{d})^{-1}.$

As regards the set $A \cap \bar{\Delta}_\eta^d$ we obtain from \eqref{C1} and \eqref{C2a} the estimate
\begin{equation}\label{C4}
\int_{A \cap \bar{\Delta}^d_\eta}[\int_{\Delta^d}(\log (\langle \pi_{Lip} (x), \frac{y}{x}\rangle ) - \log (\langle \pi_\epsilon(x), \frac{y}{x}\rangle) d \rho(x,y)] d\rho(x) \geq - \epsilon
\end{equation}
and a similar estimate holds true for the set $\Delta^d \backslash \bar{\Delta}_\eta^d$ by \eqref{C2b}. Hence, we obtain from \eqref{C2c}, \eqref{C3}, and \eqref{C4}
$$L^{\pi_{Lip}} \geq L + \log (1-\epsilon) -d^2 \epsilon - 2 \epsilon.$$

As $\epsilon >0$ is arbitrary, we have proved Lemma \ref{5.4}.
\end{proof}

\subsection{Asymptotically equivalent growth rates}

Having introduced the three different approaches, we are now ready to compare their asymptotic performance. We first establish an analog to Cover's Theorem \ref{t1.1} in the present setting.

\begin{theorem}\label{t4.6}
Fix a Borel probability measure $\nu$ with full support on $\mathcal{L}^M$ as in \eqref{B1} above.
For every individual sequence $(\mu_t)^\infty_{t=0}$ in $\Delta^d$ we have 
\begin{align}\label{p9}
\lim_{T \to\infty} \frac{1}{T} (\log(V_T^{*,M})-\log(V^M_T(\nu)))=0.
\end{align}
\end{theorem}

\begin{proof}

The inequality ``$\geq$'' is obvious. For the reverse inequality we follow the argument of \cite{blum1999universal}. Since ${\mathcal{L}}^M$ is compact and $\nu$ has full support, it is not difficult to see that for any $\eta >0$, there exists $\delta>0$ such that every $\eta$-neighbourhood of a point $\pi \in \mathcal{L}^M$ has $\nu$-measure bigger than $\delta.$

Let a scenario $(\mu_t)_{t = 0}^{\infty}$ in $\Delta^d$ be given. For a fixed time $T$ let $\pi^{*,M}\in \mathcal{L}^M$ be an optimizer of \eqref{p7}. Consider a portfolio map $\pi^M \in \mathcal{L}^M$ with $\|\pi^M-\pi^{*,M}\|_\infty < \eta,$ i.e., such that, for every $x \in \Delta^d$ we have $\|\pi^M(x) - \pi^{*,M}(x)\|_1 = \sum^d_{j=1} |\pi^M(x)^j-\pi^{*,M}(x)^j |< \eta.$

Choose $\eta>0$ small enough so that $\alpha=\eta Md < 1$ and define, for $x \in \Delta^d,$ 
\begin{align}\label{p8}
\tilde{\pi}(x)=\frac{1}{\alpha} \pi^M(x) - \frac{1-\alpha}{\alpha} \pi ^{*,M}(x),
\end{align}
so that
\begin{align}\label{p8a}
\pi^M(x)=(1-\alpha)\pi^{*,M}(x)+\alpha \tilde{\pi}(x).
\end{align}

Observe that we have that $\tilde{\pi}$ maps $\Delta^d$ into $\bar{\Delta}^d.$ Indeed, fix $1\leq j \leq d.$ As $|\pi^M(x)^j -  \pi^{*,M}(x)^j| < \eta$ we obtain
\begin{equation}
    \begin{split}
\\ \tilde{\pi} (x)^j &= \pi^{*,M}(x)^j + \frac{1}{\alpha}(\pi^M(x)^j-\pi^{*,M}(x)^j)
\\ &\geq (M d)^{-1} -\frac{\eta}{\alpha} = 0.
\end{split}
\end{equation}

Using \eqref{p8a} we may estimate 
\begin{equation}
    \begin{split}
\\ \frac{1}{T} \log V^{\pi^M}_T&=\frac{1}{T} \sum^{T-1}_{t=0} \log (\langle \pi^M(\mu_t), \frac{\mu_{t+1}}{\mu_t}\rangle)
\\& \geq \frac{1}{T} \sum^{T-1}_{t=0} \log (\langle (1-\alpha) \pi^{*,M}(\mu_t), \frac{\mu_{t+1}}{\mu_t}\rangle)
\\& = \frac{1}{T} \log(V_T^{*,M})+ \log(1-\alpha).     
\end{split}
\end{equation}

Fix $\epsilon >0.$ Choosing $\eta >0$ sufficiently small we can make $\alpha =\eta M d$ small enough such that the final term is bigger than $-\epsilon$.
Summing up, we have
\begin{align}\label{9p}
\frac{1}{T}[\log (V^{*,M}_T) - \log (V_T^{\pi^M})] < \epsilon
\end{align}
whenever $\| \pi^M - \pi^{*,M}\|_\infty <\eta.$

Denote by $B=B_\eta (\pi^{*,M})$ the $\|\cdot \|_\infty$-ball with radius $\eta$ in $\mathcal{L}^M$ which has $\nu$-measure at least $\delta >0,$ where $\delta$ only depends on $\eta.$ As each element $\pi^M$ of $B$ satisfies \eqref{9p} we have
\begin{align}\label{eq:compuniversal}
\frac{1}{T} \log (V_T^M(\nu)) \geq \frac{\log(\delta)}{T} + \frac{1}{T} \log (V^{*,M}_T)-\epsilon.
\end{align}
Now \eqref{p9} is proved by sending in \eqref{eq:compuniversal} $T$ to infinity and letting $\epsilon$ to zero.
\end{proof}

It is worth mentioning that in Theorem \ref{t4.6} we did not have to impose the uniform boundedness condition \eqref{p3} (compare this result with \cite[Lemma 3.3]{wong2015universal}). Now we can combine Lemma \ref{5.4}, which pertains to the probabilistic setting of this section, with Theorem \ref{t4.6} which is a result holding true for {\it every} (as opposed to {\it $\mathbb{P}$-almost every}) sequence $(\mu_t)^\infty_{t=0}$ in $\Delta^d,$ to obtain -- under suitable assumptions --  equality of the asymptotic performance among the three classes of strategies introduced in Section \ref{sec:portfoliodiscrete}.

We first state the equivalence of the performance properties in the context of a fixed Lipschitz constant $M$. In Corollary \ref{c4.8} below we then formulate a qualitative result not relying on this constant.

\begin{theorem}\label{tG3}
Let $\mu=(\mu_t)^\infty_{t=0}$ be a $\Delta^d$-valued stochastic process satisfying Assumptions \ref{5.1} and \ref{5.2} and being defined on the canonical space $\Omega =(\Delta^d)^\mathbb{N}$ equipped with its natural filtration and a probability measure $\mathbb{P}.$ Let $M>0$ be a fixed Lipschitz constant. 
\begin{enumerate}
\item [(i)]
For a random element $\omega \in \Omega$ consider the trajectory $(\mu_t(\omega))^\infty_{t=0}.$ For each $T \in \mathbb{N},$ define the element $\pi^{*,M}(\omega) \in \mathcal{L}^M$ as well as the positive real numbers $V^{*,M}_T(\omega) :=V^{\pi^{*,M}}_T(\omega)$ as in \eqref{p7}.
\item [(ii)]
Fix a probability measure $\nu$ on $\mathcal{L}^M$ with full support and denote by $V^M(\nu)(\omega),$ for a given trajectory $(\mu_t(\omega))^\infty_{t=0},$ the corresponding sequence of positive real numbers $(V^M_t(\nu)(\omega))^\infty_{t=0}$ as in \eqref{B1}.
\item [(iii)]
Define the log-optimal wealth process among the portfolio maps $\pi \in \mathcal{L}^M$ by 
\begin{align}
\hat{\pi}^M = \argmax_{\pi \in \mathcal{L}^M} \int_{\Delta^d} \left[\int_{\Delta^d} \log (\langle \pi (x), \frac{y}{x} \rangle ) \rho(x,dy) \right] d\rho (x)
\end{align}
and define the corresponding wealth process $(\hat{V}^M_t(\omega))^\infty_{t=0}=(V^{\hat{\pi}^M}_t(\omega))^\infty_{t=0}$ as in Theorem \ref{tG3}.
\end{enumerate}
Then, for $\mathbb{P}$-almost all trajectory $(\mu_t(\omega))^\infty_{t=0}$, we have
\begin{align}\label{G4}
\liminf_{T\to \infty} \frac{1}{T} \log ( V^{*,M}_T(\omega ) ) = \liminf_{T\to \infty} \frac{1}{T} \log (V^M_T(\nu) (\omega))=
\lim_{T \to \infty} \frac{1}{T} \log (\hat{V}^M_T).
\end{align}
In addition, the first equality holds for {\it all} sequences $(\mu_t(\omega))^\infty_{t=0}$ in $\Delta^d.$ 
\end{theorem}

\begin{proof}
Regarding the well-definedness, the only issue which has not been shown already is the well-definedness of $\hat{\pi}^M$. 
This fact follows in a straight-forward way from the compactness of $\mathcal{L}^M$ with respect to $\|\cdot\|_{\infty}$ (compare the proof of Lemma \ref{5.4}).
Note also that by the ergodic theorem (Theorem \ref{t4.3}), for each $\pi \in {\mathcal{L}}^M$ we have the a.s.~limit
\[
\lim_{T \rightarrow \infty} \frac{1}{T} \log V_T^{\pi}(\omega) = L^{\pi}.
\]
Here we recall that $L^{\pi}$ is defined by \eqref{E4a}. In particular, since $\hat{\pi}^M \in {\mathcal{L}}^M$ by definition, we have a.s.
\begin{align}\label{eq:LM}
\lim_{T \rightarrow \infty} \frac{1}{T} \log \hat{V}^M_T(\omega) = \sup_{\pi \in {\mathcal{L}}^M} L^{\pi}=: L^M.
\end{align}

That the first equality in \eqref{G4} holds for {\it all} sequences $(\mu_t(\omega))^\infty_{t=0}$ in $\Delta^d$, was shown in Theorem \ref{t4.6}.

For each fixed $T \in \mathbb{N}$ and each scenario $(\mu_t(\omega))^T_{t=0},$ we have the obvious a.s.~inequality
\begin{align}\label{G5a}
\frac{1}{T} \log (\hat{V}^M_T) \leq \frac{1}{T} \log (V^{*,M}_T(\omega)).
\end{align}
Indeed, \eqref{G5a} follows from the fact that $\pi^{*,M}_T \in \mathcal{L}^M$ is chosen 
retrospectively as the best performing element in $\mathcal{L}^M$. It therefore dominates the choice $\hat{\pi}^M \in \mathcal{L}^M$. Recall that the latter has to be made without knowing the trajectory $(\mu_t(\omega))^T_{t=0}.$

Using \eqref{eq:LM}, \eqref{G5a} and Theorem \ref{t4.6} we thus have for $\mathbb{P}$-a.e. scenario $(\mu_t(\omega))_{t=0}^T$,
\begin{align}\label{eq:pathwiselim}
L^M= \lim_{T\to \infty}\frac{1}{T} \log (\hat{V}^M_T) \leq \liminf_{T \to \infty} \frac{1}{T} \log (V^{*,M}_T(\omega))=\liminf_{T\to \infty} \frac{1}{T} \log (V^M_T(\nu) (\omega)).
\end{align}
On the other hand, by the definition of $(\hat{V}^M_t)^\infty_{t=0}$ as the log-optimizer within the class $\mathcal{L}^M$, we have
\begin{align}\label{eq:expectdom}
 \mathbb{E}[\log (V^M_T(\nu)) ]\leq \sup_{\pi \in \mathcal{L}^M}\mathbb{E}[\log (V^{\pi}_T)]=\mathbb{E}[\log (\hat{V}^M_T)].
\end{align}
Note here that the universal portfolio process to build $V^{M}_T(\nu)$ is at each time a random mixture of Lipschitz maps lying in $\mathcal{L}^M$. 
By the time-homogenous Markovianity it is thus sufficient to dominate the left hand side of \eqref{eq:expectdom} by taking the supremum over elements in $\mathcal{L}^M$.
Combining now \eqref{eq:expectdom}, 
Theorem \ref{t4.3} and \eqref{eq:pathwiselim} yields,
\begin{align*}
 \mathbb{E}\left[\liminf_{T\to \infty}\frac{1}{T} \log (V^M_T(\nu))\right]&\leq \liminf_{T\to \infty} \frac{1}{T} \mathbb{E}[\log (V^M_T(\nu)) ]\\
 &\leq \lim_{T\to \infty} \frac{1}{T} \mathbb{E}[\log (\hat{V}^M_T)]\\
 &= \lim_{T\to \infty}\frac{1}{T} \log (\hat{V}^M_T) \\
 &\leq \liminf_{T\to \infty} \frac{1}{T} \log (V^M_T(\nu)), \quad \mathbb{P}\text{-a.s.},
\end{align*}
where the first inequality follows from Fatou's lemma (note here that $\log (V^M_T(\nu))$ is bounded from below, see e.g., \eqref{eq:compuniversal}).
From this we see that 
\[
\liminf_{T\to \infty} \frac{1}{T} \log (V^M_T(\nu))
\]
is $\mathbb{P}$-almost surely constant and equal to 
$\lim_{T\to \infty}\frac{1}{T} \log (\hat{V}^M_T) $. Hence the assertion is proved.

\end{proof}

To formulate a result not depending explicitly on the constant $M$, we define a universal portfolio process $(V_t(\nu))^\infty_{t=0}$ in the following way. For $M=1,2,3, \dots$ choose a measure $\nu^M$ on $\mathcal{L}^M$ with full support. Define $\nu=\sum^{\infty}_{M=1} 2^{-M} \nu^M$ and the process $V(\nu)$ as in \eqref{B1} 
\begin{align}\label{H1}
V_t(\nu)=\int_{\bigcup^\infty_{M=1}\mathcal{L}^M} V^\pi_t d\nu(\pi), \qquad t \in \mathbb{N}.
\end{align}
To be precise, here we consider $\bigcup_{M = 1}^{\infty} {\mathcal{L}}^M$ as a disjoint union. Recall that $\hat{V}_t$ is the wealth process of the log-optimal portfolio \eqref{E2}. 

\begin{corollary}\label{c4.8}
Under the assumptions of Theorem \ref{tG3} we have $\mathbb{P}$-a.s.
\begin{align}
\lim_{M \rightarrow \infty} \lim_{T \rightarrow \infty} \frac{1}{T} \log  V_T^{*, M}=\lim_{T\to \infty} \frac{1}{T} \log V_T(\nu) = \lim_{T\to \infty} \frac{1}{T} \log \hat{V}_T =  L,
\end{align}
where $L$ is defined in \eqref{eq:ass4.2}.
\end{corollary}

\begin{proof}
Note that 
\begin{align*}
\lim_{M \rightarrow \infty} \liminf_{T \rightarrow \infty} \frac{1}{T} \log  V_T^{*, M} =\liminf_{T\to \infty} \frac{1}{T} \log V_T(\nu) = \lim_{T\to \infty} \frac{1}{T} \log \hat{V}_T = L
\end{align*} 
follows from  Lemma \ref{5.4} and Theorem \ref{tG3}. Since Theorem \ref{t4.6} implies that
\[
\limsup_{T\to \infty} \frac{1}{T} \log V_T(\nu) = \lim_{M \rightarrow \infty} \limsup_{T \rightarrow \infty} \frac{1}{T} \log  V_T^{*, M},
\]
the corollary is proved if 
\begin{align}\label{eq:finite}
\limsup_{T\to \infty} \frac{1}{T} (\log V_T(\nu)-\log \hat{V}_T) = \limsup_{T\to \infty} \frac{1}{T} \log \left(\frac{V_T(\nu)} {\hat{V}_T}\right) =0, \quad \mathbb{P}\text{-a.s.}
\end{align}
As by Lemma \ref{lem:supermart}, $(\frac{V_t(\nu)}{\hat{V}_t})_{t =0}^{\infty}$ 
is a non-negative supermartingale, it converges $\mathbb{P}$-a.s. to a finite limit as $t \to \infty$. This in turn implies \eqref{eq:finite} and proves the assertion.
\end{proof}

\begin{lemma}\label{lem:supermart}
The process $(\frac{V_t(\nu)}{\hat{V}_t})_{t =0}^{\infty}$ is a non-negative supermartingale.
\end{lemma}

\begin{proof}
First note that for any $\pi : \Delta^d \to \bar{\Delta}^d $, $(\frac{V^\pi_t}{\hat{V}_t})_{t =0}^{\infty}$ is a non-negative supermartingale. Indeed, by Lemma \ref{lem:leq1}
\begin{align*}
 E\left[ \frac{V^\pi_{t+1}}{\hat{V}_{t+1}} \big | \mathcal{F}_{t}\right] =\frac{V^\pi_{t}}{\hat{V}_{t}} \int_{\Delta^d} \frac{\langle \pi(\mu_{t}), \frac{y}{\mu_{t}}\rangle}{\langle \hat{\pi}(\mu_{t}), \frac{y}{\mu_{t}}\rangle} \rho(\mu_{t},dy) \leq \frac{V^\pi_{t}}{\hat{V}_{t}}.
\end{align*}
By Fubini's theorem we can thus conclude the supermartingale property of $(\frac{V_t(\nu)}{\hat{V}_t})_{t =0}^{\infty}$:
\begin{align*}
E\left[\frac{V_{t+1}(\nu)}{\hat{V}_{t+1}}\big | \mathcal{F}_{t}\right]&=
E\left[\int_{\bigcup^\infty_{M=1}\mathcal{L}^M} \frac{V^\pi_{t+1}}{\hat{V}_{t+1}} d\nu(\pi)\big | \mathcal{F}_{t}\right]\\
&=\int_{\bigcup^\infty_{M=1}\mathcal{L}^M} E\left[ \frac{V^\pi_{t+1}}{\hat{V}_{t+1}} \big | \mathcal{F}_{t}\right] d\nu(\pi)\\
&\leq \int_{\bigcup^\infty_{M=1}\mathcal{L}^M}  \frac{V^\pi_{t}}{\hat{V}_{t}}d\nu(\pi)=\frac{V_{t}(\nu)}{\hat{V}_{t}}.
\end{align*}
\end{proof}

The following lemma is needed to establish the supermartingale property stated in Lemma \ref{lem:supermart}.

\begin{lemma}\label{lem:leq1}
Let $\hat{\pi}$ be given by \eqref{E2}. Then for every $\pi : \Delta^d \to \bar{\Delta}^d $ and every $x \in \Delta^d$,
\[
\int_{\Delta^d} \frac{\langle \pi(x), \frac{y}{x}\rangle}{\langle \hat{\pi}(x), \frac{y}{x}\rangle} \rho(x,dy) \leq 1.
\]
\end{lemma}

\begin{proof}
We proceed as in the proof of \cite[Proposition 4.3]{becherer2001numeraire}. Fix $\pi$ and $\alpha \in (0,1)$ and define $\pi^{\alpha}=\alpha \pi+ (1-\alpha) \hat{\pi}$. Then by the (long only) log-optimality of $\hat{\pi}$ we have for every $x \in \Delta^d$
\begin{align*}
0&\leq \int_{\Delta^d} \left(\log\langle \hat{\pi}(x), \frac{y}{x}\rangle-  \log\langle \pi^{\alpha}(x), \frac{y}{x}\rangle\right) \rho(x, dy)= \int_{\Delta^d} \left(\int^{\langle \hat{\pi}(x), \frac{y}{x}\rangle}_{\langle \pi^{\alpha}(x), \frac{y}{x}\rangle} \frac{1}{z} dz \right) \rho(x,dy)\\
& \leq \int_{\Delta^d} \frac{\langle \hat{\pi}(x), \frac{y}{x}\rangle -\langle \pi^{\alpha}(x), \frac{y}{x}\rangle}{\langle \pi^{\alpha}(x), \frac{y}{x}\rangle}\rho(x,dy)=\int_{\Delta^d} \frac{\langle\alpha (\hat{\pi}(x)- \pi(x)),\frac{y}{x}\rangle}{\langle \pi^{\alpha}(x), \frac{y}{x}\rangle} \rho(x,dy).
\end{align*}
Hence,
\[
\int_{\Delta^d} \frac{\langle\pi(x),\frac{y}{x}\rangle}{\langle \pi^{\alpha}(x), \frac{y}{x}\rangle} \rho(x,dy) \leq \int_{\Delta^d} \frac{\langle \hat{\pi}(x),\frac{y}{x}\rangle}{\langle \pi^{\alpha}(x), \frac{y}{x}\rangle} \rho(x,dy) \leq 
\int_{\Delta^d} \frac{\langle \hat{\pi}(x),\frac{y}{x}\rangle}{\langle (1-\alpha)\hat{\pi}(x), \frac{y}{x}\rangle} \rho(x,dy), 
\]
where the last equality follows from $\pi^{\alpha} \geq (1-\alpha)\hat{\pi}$.
By Fatou's lemma we therefore have 
\begin{align*}
\int_{\Delta^d} \frac{\langle\pi(x),\frac{y}{x}\rangle}{\langle \hat{\pi}(x), \frac{y}{x}\rangle} \rho(x,dy)  &=\int_{\Delta^d} \lim_{\alpha \to 0}\frac{\langle\pi(x),\frac{y}{x}\rangle}{\langle \pi^{\alpha}(x), \frac{y}{x}\rangle} \rho(x,dy) \leq \lim_{\alpha \to 0} \int_{\Delta^d} \frac{\langle\pi(x),\frac{y}{x}\rangle}{\langle \pi^{\alpha}(x), \frac{y}{x}\rangle} \rho(x,dy)\\
&\leq \lim_{\alpha \to 0}\frac{1}{1-\alpha}\int_{\Delta^d} \frac{\langle \hat{\pi}(x),\frac{y}{x}\rangle}{\langle \hat{\pi}(x), \frac{y}{x}\rangle} \rho(x,dy)=1,
\end{align*}
which proves the claim.
\end{proof}

\section{A comparison of the three approaches - the continuous time case with functionally generated portfolios}\label{sec:functionally}

This part of the paper is dedicated to a similar analysis as in the previous section, 
however in continuous time and with a different class of portfolios, namely \emph{functionally generated} ones, 
where we can -- even in continuous time -- treat the universal and the best retrospectively chosen portfolio in a pathwise way. 
Indeed, in the present context this class of portfolios is essentially the most general setting which allows 
to define wealth processes in a pathwise way without passing to rough paths theory. This is due to the fact that we can apply the pathwise 
It\^o-calculus developed by H.~F\"ollmer \cite{F:81} to make sense of the integral with respect to a continuous paths 
which admits a quadratic variation process in a sense made precise below. For these reasons we introduce instead of Lipschitz portfolio maps, 
the following spaces of functionally generated portfolio maps.

\subsection{Functionally generated portfolios}\label{sec:funcgen}

We consider the following set of concave functions for some fixed $M >0$ and $0 \leq \alpha \leq 1$,
\[
\mathcal{G}^{M,\alpha}=\left\{ G \in  C^{2, \alpha}(\bar{\Delta}^d), \textrm{concave such that } \| G\|_{C^{2,\alpha}} \leq M \textrm{ and } G \geq \frac{1}{M}\right\},
\]
where $ C^{2, \alpha}(\bar{\Delta}^d) $ denotes the H\"older space of $2$-times continuously differentiable functions from $\bar\Delta^d \to \mathbb{R}$ whose derivatives are $\alpha$-H\"older continuous. That is,
\[
C^{2, \alpha}(\bar{\Delta}^d)=\{ G \in C^2(\bar{\Delta}^d) \, |\, \| G\|_{C^{2,\alpha}} < \infty\},
\]
where 
\[
\| G\|_{C^{2,\alpha}}= \max_{|\mathbf{k}| \leq 2} \| D^{\mathbf{k}} G \|_{\infty} + 
\max_{|\mathbf{k}|=2} \sup_{x \neq y}\frac{|D^{\mathbf{k}} G(x)-D^{\mathbf{k}} G(y)|}{\|x-y\|^{\alpha}}
\]
with $\mathbf{k}$ denoting a multiindex in $\mathbb{N}^2$. For $\alpha=0$ the second term in this norm is left away. Note that $G$ is only defined on the simplex $\Delta^d$. In order that the partial derivatives are well defined, we assume that each $G$ is extended to an open neighborhood of $\Delta^d$ such that $G(x) = G(x')$, where $x'$ is the orthogonal projection of $x$ onto $\Delta^d$. The choice of the extension is immaterial and is for notational convenience only.

\begin{lemma}\label{lem:compact}
Let $M, \alpha >0$ be fixed. Then the set $\mathcal{G}^{M,\alpha}$ is compact with respect to $\| \cdot\|_{C^{2,0}}$.
\end{lemma}

\begin{proof}
This follows from the fact that the embedding from $C^{2,\alpha}(\bar{\Delta}^d) \to C^{2, \alpha'}(\bar{\Delta}^d)$ is compact for $\alpha' < \alpha$ (see e.g., \cite[Satz 2.42]{D:10}). This means in particular that any bounded set in $C^{2,\alpha}(\bar{\Delta}^d) $ is totally bounded in $C^{2,0}(\bar{\Delta}^d)$, thus relatively compact. To prove compactness it thus suffices to prove that $\mathcal{G}^{M,\alpha}$ is closed. Take a sequence $G^n\in \mathcal{G}^{M,\alpha}$ converging to $G$ with respect to the $\|\cdot\|_{C^{2,0}}$ norm. Then, we can estimate $\| G\|_{C^{2,\alpha}}$ by
\begin{align*}
&\| G\|_{C^{2,\alpha}} = \| G\|_{C^{2,0}}+ 
\max_{|\mathbf{k}|=2} \sup_{x \neq y}\frac{|D^{\mathbf{k}} G(x)-D^{\mathbf{k}} G(y)|}{\|x-y\|^{\alpha}}\\
&\leq \|G-G^n\|_{C^{2,0}}+ \|G^n\|_{C^{2,0}}\\
&\quad + \max_{|\mathbf{k}|=2} \sup_{x \neq y}\frac{|D^{\mathbf{k}} G(x)
-D^{\mathbf{k}}G^n(x)| +|D^{\mathbf{k}}G^n(x) -D^{\mathbf{k}} G^n(y)| + |D^{\mathbf{k}} G^n(y)
-D^{\mathbf{k}} G(y)|}{\|x-y\|^{\alpha}}
\end{align*}
for any $n \in \mathbb{N}$. Letting $n \to \infty$ and using the fact that $ \|G^n -G\|_{C^{2,0}} \to 0$  yields $\|G\|_{C^{2,\alpha}} \leq M$. Similarly, we obtain 
$G \geq \frac{1}{M}$. This together with the fact that $G$ is concave as a limit of concave functions proves $G \in \mathcal{G}^{M,\alpha}$ and thus in turn compactness of $\mathcal{G}^{M,\alpha}$ with respect to $\| \cdot\|_{C^{2,0}}$.
\end{proof}

To the set of generating functions $\mathcal{G}^{M,\alpha}$ we associate now the set of functionally generated portfolios $\mathcal{FG}^{M,\alpha}$
in the spirit of \cite{F:02} defined by
\begin{equation}\label{eq:FG}
\begin{split}
&\mathcal{FG}^{M,\alpha}=\Bigg\{ \pi^G : \Delta^d \to \bar{\Delta}^d, \\
&\quad x \mapsto (\pi^G(x))^i= x^i\left(\frac{D^i G(x)}{G(x)}+1 -\sum_{j=1}^d \frac{D^j G(x)}{G(x)} x^j\right), \, i =1, \ldots d, \, |\, G \in \mathcal{G}^{M, \alpha}\Bigg\}.
\end{split}
\end{equation}
Note that due to the concavity of $G$, $\pi^G$ takes values in $\bar{\Delta}^d$, i.e. it is long-only.
We will denote the corresponding portfolio wealth processes either by $V^{\pi^G}$ or $V^{G}$.

As mentioned in the introduction, for these kinds of portfolios it is possible to obtain -- via the pathwise It\^o-calculus developed by F\"ollmer~\cite{F:81} --  
a pathwise expression for $V^{\pi^G}$. 
In the specific context of stochastic portfolio theory this has recently been worked out by A.~Schied, L.~Speiser and I.~Voloshchenko \cite{SSV:16}. 
Within this pathwise approach they can also include more general functionally generated portfolios, 
whose generating function can depend on time as well as the paths of the market weights or on further continuous trajectories of bounded variation in an 
adaptive manner. This is achieved by applying the functional It\^o-calculus developed by B.~Dupire \cite{D:09} and R.~Cont and D.~Fourni\'e \cite{CF:10, CF:13}, 
which generalizes F\"ollmer's It\^o-calculus to path-dependent functionals. 
In order to be in line with the previous section we only consider ``Markovian'' strategies, i.e., functionally generated portfolios as defined in \eqref{eq:FG}.

We adopt the notation of \cite{SSV:16} and fix a refining sequence of partitions $(\mathbb{T}_n)_{n=1}^ {\infty}$ of $[0,\infty)$, i.e., $\mathbb{T}_n =\{ t_0, t_1,\ldots\}$ is such that $0 = t^n_0 < t^n_1 < \cdots$ and $t^n_k \to \infty$ as $k \to \infty$, and $\mathbb{T}_1 \subset \mathbb{T}_2\subset \cdots$. Moreover, the mesh of $\mathbb{T}_n$ tends to zero on each compact interval as $n \to \infty$. Furthermore, we denote the successor of $t \in  \mathbb{T}_n$ by $t'$. That is,
$t' = \min\{ u \in \mathbb{T}_n \, |,\ u > t\}$. 

Throughout this section the market weights are described by a $d$-dimensional continuous path $\mu=(\mu_t)_{t \geq 0}$ with values in $\Delta^d$. Furthermore, $\mu$ has to exhibit pathwise quadratic variation in an appropriate sense. Here and henceforth we let $S_d^+$ be the set of $d \times d$ positive definite matrices.

 \begin{assumption}\label{ass:quadratic}
The path $(\mu_t)_{t \geq 0}$ admits a continuous $S_d^+$-valued quadratic variation $[\mu]$ along $(\mathbb{T}_n)$ in the sense of \cite{F:81}, 
i.e., for any $1 \leq i, j \leq d$ and all $ t\geq 0$ the sequence
\[
\sum_{s \in \mathbb{T}_n, s \leq t} (\mu^i_{s'}-\mu^i_s)(\mu^j_{s'}-\mu^j_s)
\]
converges to a finite limit, as $n \to \infty$, denoted $[\mu^i, \mu^j]_t$, such that $t \mapsto [\mu^i, \mu^j]_t$ is continuous.
\end{assumption}

The dynamics of the relative wealth process $V^{\pi^G}$ built by investing according to $\pi^G \in \mathcal{FG}^{M,\alpha}$ 
are given in this continuous time case by  
\begin{align}\label{eq:wealthfunc}
\frac{dV_t^{\pi^G}}{V_{t}^{\pi^G}}=\sum_{i=1}^d (\pi^G(\mu_t))^i\frac{d\mu_{t}^i}{\mu_{t}^i} =\sum_{i=1}^d \frac{D^i G(\mu_t)}{G(\mu_t)}d\mu_{t}^i, 
\quad V_0^{\pi} =1,
\end{align}
(compare \eqref{p4} in the discrete time case), where the right hand side has to be understood as F\"ollmer's pathwise integral (c.f.~ Equation (2.5) in \cite{SSV:16}).
Note that the second equality holds by the definition of $\pi^G$ and the fact that $\sum_{i=1}^d d\mu_{t}^i=0$.

Using \eqref{eq:wealthfunc} and F\"ollmer's It\^o calculus, we then get the following pathwise version of
Fernholz's \cite{F:02} master equation, which also follows from \cite[Theorem 2.9]{SSV:16}.

\begin{corollary}
Let $G \in C^2(\bar{\Delta}^d)$ and $\pi^G$ be defined as in \eqref{eq:FG}. Let $(\mu_t)_{t \geq 0}$ be a continuous path satisfying Assumption \ref{ass:quadratic}. 
Then the relative wealth process $V^{\pi^G}$ satisfies
\begin{align}\label{eq:master}
V_T^{\pi^G}\equiv V_T^G=\frac{G(\mu_T)}{G(\mu_0)}e^{\mathfrak{g}([0,T])}, \quad 0 \leq T < \infty,
\end{align}
where $\mathfrak{g}(dt)=-\frac{1}{2G(\mu_t)}\sum_{i,j} D^{ij} G(\mu_t) d[\mu^i, \mu^j]_t$.
\end{corollary}

\subsection{Different types of portfolios}\label{sec:typesfun}

In this continuous-time setting we consider again the following three types of portfolios
\begin{enumerate}
\item [(i)] the best retrospectively chosen portfolio,
\item [(ii)] the universal portfolio and,
\item [(iii)] the log-optimal portfolio, 
\end{enumerate}
now all within the class of functionally generated portfolios. As before our goal is to establish equality of the asymptotic performance. 
To define the log-optimal portfolio we will need to pass to a specific stochastic model which we introduce in
Section \ref{sec:logoptimalfunc}. In Section \ref{sec:assgrowth} we then establish the form of the asymptotic growth rate 
for this model class under an additional ergodicity assumption.

\subsubsection{The best retrospectively chosen portfolio}
As in discrete time, let us consider the best portofolio with hindsight, 
this time among the set of functionally generated portfolios $\mathcal{FG}^{M,\alpha}$ and 
for a given continuous path $(\mu_t)_{t \geq 0}$ satisfying Assumption \ref{ass:quadratic}. Fix $M, \alpha >0$ and define
\begin{align}\label{eq:retrofunctional}
V^{*,M,\alpha}_T= \sup_{ \pi^G \in \mathcal{FG}^{M,\alpha}}V_T^{\pi^G}= \sup_{ G \in \mathcal{G}^{M,\alpha}}V_T^{G}.
\end{align}

To prove existence of an optimizer, let us start by establishing continuity of $G \mapsto V^G_T$:

\begin{lemma}\label{lem:continuity}
Let $T, M, \alpha >0$ be fixed and $(\mu_t)_{t \geq 0}$ be a continuous path satisfying Assumption \ref{ass:quadratic}.
Consider the function $G \mapsto V^G_T $ where $V^G_T$ is given by \eqref{eq:master}. Then 
$G \mapsto V^G_T $ is continuous from $(\mathcal{G}^{M,\alpha}, \| \cdot \|_{C^{2,0}})$ to $\mathbb{R}$.
\end{lemma}

\begin{proof}
For $G, \widetilde{G} \in \mathcal{G}^{M,\alpha} $, we have
\begin{align*}
&\log (V_T^{G})-\log (V_T^{\widetilde{G}})= \log( G(\mu_T))- \log(\widetilde{G}(\mu_T))-
(\log( G(\mu_0))- \log(\widetilde{G}(\mu_0)))\\
&\quad\quad -\int_0^T\left(\sum_{i,j} \frac{D^{ij} G(\mu_t)}{2G(\mu_t)}-\frac{D^{ij} \widetilde{G}(\mu_t)}{2\widetilde{G}(\mu_t)}\right) d[\mu^i, \mu^j]_t\\
&\quad =
\log( G(\mu_T))- \log(\widetilde{G}(\mu_T))-
(\log( G(\mu_0))- \log(\widetilde{G}(\mu_0)))\\
&\quad \quad-\sum_{i,j} \int_0^T\left(\frac{D^{ij} G(\mu_t)- D^{ij} \widetilde{G}(\mu_t)}{2G(\mu_t)}+\frac{D^{ij} \widetilde{G}(\mu_t)\left(\widetilde{G}(\mu_t)-G(\mu_t)\right)}{2\widetilde{G}(\mu_t)G(\mu_t)}\right) d[\mu^i, \mu^j]_t.
\end{align*}
Hence, using the fact that $\|\widetilde{G}\|_{C^{2,0}} \leq M$ as well as $G \geq \frac{1}{M}$ and $\widetilde{G} \geq \frac{1}{M}$ and that $z \mapsto \log(z)$ is Lipschitz continuous on $[\frac{1}{M}, \infty)$ with constant $M$, we obtain the following estimate:
\begin{equation}
\begin{split}\label{eq:funcont1}
|\log (V_T^{G})-\log (V_T^{\widetilde{G}})|&\leq 2 M\|G- \widetilde{G}\|_{C^{2,0}}\\
&\quad +\left(\frac{M}{2}d^2 \|G- \widetilde{G}\|_{C^{2,0}}+
\frac{M^3}{2} d^2\|G-\widetilde{G}\|_{C^{2,0}}\right)\max_i[\mu^i, \mu^i]_T,
\end{split}
\end{equation}
which proves the asserted continuity.
\end{proof}

As a consequence we obtain the existence of an optimizer as stated in the following proposition.

\begin{proposition}\label{prop:optimizer}
Let $T$ be fixed and $(\mu_t)_{t\geq 0}$ be a continuous path satisfying Assumption \ref{ass:quadratic}. Let $V^{*,M,\alpha}_T$ be defined by \eqref{eq:retrofunctional}. Then there exists an optimizer $G_T^{*,M,\alpha}\in \mathcal G^{M,\alpha} $and in turn a portfolio $\pi_T^{*,M,\alpha}$ generated by $G_T^{*,M,\alpha}$ such that
\[
V^{*,M,\alpha}_T=V_T^{\pi_T^{*,M,\alpha}}=V^{G_T^{*,M,\alpha}}_T.
\]
\end{proposition}

\begin{proof}
This is simply a consequence of continuity as proved in Lemma \ref{lem:continuity} and
compactness of $(\mathcal{G}^{M,\alpha}, \|\cdot\|_{C^{2,0}})$ as shown in Lemma \ref{lem:compact}.
\end{proof}

\subsubsection{Universal portfolio}
In order to define the analog of Cover's portfolio in the present setting, let $m$ be a Borel probability measure on $(\mathcal{G}^{M,\alpha}, \|\cdot\|_{C^{2, 0}})$. Consider the map
\begin{align}\label{eq:Pi}
F: \mathcal{G}^{M,\alpha} \to \mathcal{FG}^{M, \alpha};\, G \mapsto F(G)=\pi^G,
\end{align}
where $\pi^G$ is given by \eqref{eq:FG}. Define now on $(\mathcal{FG}^{M,\alpha}, \|\cdot\|_{\infty})$ a Borel probability measure $\nu$ via $\nu=F_{*} m$ (the pushforward) and set
\[
\pi^{M, \alpha,\nu}_T=\frac{\int_{\mathcal{FG}^{M,\alpha}} \pi^G(\mu_T) V_T^{\pi^G}d\nu(\pi^G)}{\int_{\mathcal{FG}^{M,\alpha} } V_T^{\pi^G} d\nu(\pi^G)}.
\]
Analogous to \eqref{C8} and \eqref{B1} we have the following form for the universal portfolio 
(compare also \cite{J:92} and \cite{wong2015universal}).

\begin{lemma}\label{lem:universalfunc}
The relative wealth achieved by investing according to $\pi^{M,\nu,\alpha}_T$ is given by
\begin{align}\label{eq:universalwealthfunk}
V^{M,\alpha}_T(\nu):=V^{\pi^{M,\alpha,\nu}}_T=\int_{\mathcal{FG}^{M,\alpha}} V_T^{\pi^G}d\nu(\pi^G)= \int_{\mathcal{G}^{M,\alpha}} V^G_T m(dG).
\end{align}
\end{lemma}

\begin{proof}
Note that by~\eqref{eq:wealthfunc}, we have
\begin{align*}
\frac{d\int_{\mathcal{FG}^{M,\alpha}} V_T^{\pi^G}d\nu(\pi^G)}{\int_{\mathcal{FG}^{M,\alpha}} V_T^{\pi^G}d\nu(\pi^G)}&=
\frac{\int_{\mathcal{FG}^{M,\alpha}} dV_T^{\pi^G}d\nu(\pi^G)}{\int_{\mathcal{FG}^{M,\alpha}} V_T^{\pi^G}d\nu(\pi^G)}\\
&=\frac{\int_{\mathcal{FG}^{M,\alpha}} 
\left( \sum_{i=1}^d V_T^{\pi}\pi^G(\mu_T)^i\frac{d\mu^i_{T}}{\mu^i_{T}}\right)d\nu(\pi^G)}{\int_{\mathcal{FG}^{M,\alpha}} V_T^{\pi^G}d\nu(\pi^G)}\\
&=\sum_{i=1}^d \frac{\left(\int_{\mathcal{FG}^{M,\alpha}}V_T^{\pi^G} \pi^G(\mu_T)^i d\nu(\pi^G)\right)}
{\int_{\mathcal{FG}^{M,\alpha}} V_T^{\pi^G}d\nu(\pi^G)}
\frac{d\mu^i_{T}}{\mu^i_{T}}\\
&=\sum_{i=1}^d(\pi_{T}^{M,\nu,\alpha})^i\frac{d\mu^i_{T}}{\mu^i_{T}}=\frac{dV_T^{M}(\nu)}{V_T^{M}(\nu)}.
\end{align*}
\end{proof}

\subsubsection{Functionally generated log-optimal portfolios}\label{sec:logoptimalfunc}

By definition, the notion of log-optimal portfolios requires the introduction of a probability space on which we consider a stochastic model for the market weights. We suppose that $\mu=(\mu^1_t, \ldots, \mu_t^d)_{t \geq 0}$ follows a \emph{time-homogeneous Markovian It\^o-diffusion} 
(defined on $(\Omega, \mathcal{F}, (\mathcal{F}_t)_{t\geq 0}, \mathbb{P})$) with values in $\Delta^d$ of the form 
\begin{align}\label{eq:muP}
\mu_t=\mu_0+\int_0^t  c(\mu_s)\lambda(\mu_s)ds + \int_0^t \sqrt{c(\mu_s)}dW_s, \quad \mu_0 \in \Delta^{d},
\end{align}
where $\sqrt{\cdot}$ denotes the matrix square root, $W$ a $d$-dimensional Brownian motion,  
$\lambda$ a Borel measurable function from $\Delta^d \to \mathbb{R}^d$ 
and $c$ a Borel measurable function from $\Delta^d \to S_+^d$,
satisfying 
\begin{align}\label{eq:Deltainside1}
&\int_0^T \lambda^{\top}(\mu_t)c(\mu_t) \lambda(\mu_t) dt<\infty, \quad \forall  T \in [0,\infty), \\
&c(x)\mathbf{1}=0, \qquad 
\sum_{i,j} c^{ij}(x) \lambda(x)^j=0, \quad \forall x \in \Delta^d. \label{eq:Deltainside2}
\end{align}
The requirements in \eqref{eq:Deltainside2} are necessary to guarantee  that the process $\mu$ lies in $\Delta^d$. 
Note that $(\mu_t)_{t \geq 0}$ given by \eqref{eq:muP} satisfies the so called \emph{structure condition} (see \cite{S:95}), 
which is due to  \eqref{eq:Deltainside1} and the fact that the drift part is of form $\int_0^t c(\mu_s) \lambda(\mu_s) ds$.
Let us remark that the structure condition characterizes the condition of ``No unbounded profit with bounded risk'' (NUPBR) in the case of continuous semimartingales
(see e.g., \cite{HS:10}). This is the minimal condition for economically reasonable models in continuous time
and is a usual assumption in stochastic portfolio theory. 

In this setting proportions of current (relative) wealth invested in each of the assets are described by processes $\pi$ in the following set
\begin{align}\label{eq:Piset}
\Pi=\{\pi \,  | \, H^d\text{-valued, predictable}, R\text{-integrable}\},              
\end{align}
where the process $R$ is defined componentwise by $R_t^i=\int_0^t\frac{d\mu_{s}^i}{\mu_{s}^i}$. Recall that $H^d$ denotes the hyperplane corresponding to portfolio weights that are not necessarily long-only. Note that the set $\mathcal{FG}^{M,\alpha}$ is clearly a subset of long-only strategies in $\Pi$.

Analogously to the model free setting (c.f. Equation \eqref{eq:wealthfunc}) the dynamics of the relative wealth process $V^{\pi}$ 
built by investing according to $\pi \in \Pi$ satisfy
\begin{align}\label{eq:wealth}
\frac{dV_t^{\pi}}{V_{t}^{\pi}}=\sum_{i=1}^d \pi_t^i\frac{d\mu_{t}^i}{\mu_{t}^i}, \quad V_0^{\pi} =1.
\end{align}
In contrast to Section \ref{sec:funcgen}, the right hand side is here understood as the usual stochastic integral since we deal with general integrands $\pi$.
Note that we can also write 
\begin{align}\label{eq:portfolio}
V_T^{\pi}&=\mathcal{E}((\pi \bullet R))_T=\exp\left(\int_0^{T} \left(\frac{\pi}{\mu_t}\right)^{\top}d\mu_{t}-\frac{1}{2}\int_0^{T}\left(\frac{\pi}{\mu_{t}}\right)^{\top}c(\mu_t)\frac{\pi}{\mu_{t}}
dt\right)\\
&=\exp\left(\int_0^{T} \sum_{i=1}^d \pi^i\frac{d\mu_{t}^i}{\mu_{t}^i}-\frac{1}{2}\int_0^{T}\sum_{i,j}\frac{\pi^i}{\mu_{t}^i}\frac{\pi^j}{\mu_{t}^j}c^{ij}(\mu_t)
dt\right),\notag
\end{align}
where, for two vectors $p, q \in \mathbb{R}^d$, $\frac{p}{q}$ always denotes the componentwise quotient $(\frac{p^1}{q^1}, \ldots, \frac{p^d}{q^d})$.

Let us now turn to the log-optimal num\'eraire portfolio which is defined as in~\eqref{C15}, now however in continuous time. 
As in \cite[Section 3.1]{FKa:10}, we derive the ratio of two wealth processes $V^{\pi}$ and $V^{\theta}$ for  $\pi, \theta \in \Pi$.  
Using \eqref{eq:wealth} (for the processes $\pi$ and $\theta$) and It\^{o}'s lemma, this ratio is given by
\begin{align*}
d\left(\frac{V_t^{\pi}}{V^{\theta}_t}\right)&=\frac{V_t^{\pi}}{V^{\theta}_t}\left(\frac{\pi_t}{\mu_t}-\frac{\theta_t}{\mu_t}\right)^{\top}\left(d\mu_t-c(\mu_t)\frac{\theta_t}{\mu_t}dt\right)\\
&=\frac{V_t^{\pi}}{V^{\theta}_t}\left(\frac{\pi_t}{\mu_t}-\frac{\theta_t}{\mu_t}\right)^{\top}\left(\sqrt{c(\mu_t)}dW_t+c(\mu_t)\left(\lambda(\mu_t)
-\sqrt{c(\mu_t)}\frac{\theta_t}{\mu_t}\right) dt\right).
\end{align*}
The finite variation part of the expression vanishes \emph{for every $\pi \in \Pi$} if we choose $\theta \in \Pi$ such that
\begin{align}\label{eq:scaledweights} 
c(\mu_t)\left(\frac{\theta_t}{\mu_t} -\lambda(\mu_t)\right)=0, \quad \mathbb{P}\text{-a.s. for all } t \geq 0.  
\end{align}
By passing from the scaled relative weights $\theta/ \mu$ to ordinary portfolio weights via~\cite[Equation (5)]{FKa:10}, 
the generic solution of \eqref{eq:scaledweights} which we denote by $\pi^{\text{num}}$ can be expressed by
\begin{align}\label{eq:logopt}
(\pi_t^{\text{num}})^i=\mu^i_{t}\left(\lambda^i(\mu_t)+1-\sum_{j=1}^d \mu^j_{t}\lambda^j(\mu_t)\right).
\end{align}
Note that due to \eqref{eq:scaledweights}, the ratio $V_t^{\pi}/V_t^{{\text{num}}}$ is, for any $\pi \in \Pi$, a non-negative local martingale 
and therefore a supermartingale. Hence $V^{{\text{num}}}$  yields the relative wealth process corresponding to the log-optimal portfolio (see e.g., \cite{KK:07, FKa:10}). 
Indeed, by the supermartingale property and Jensen's inequality we have
\[
 \mathbb{E}\left[\log (V_T^{\pi})- \log (V_T^{\text{num}})\right]=\mathbb{E}\left[\log \left(\frac{ V_T^{\pi}}{V_T^{\text{num}}}\right)\right]\leq 
 \log\left(\mathbb{E}\left[\frac{ V_T^{\pi}} {V_T^{\text{num}}}\right]\right)\leq 0,
\]
showing that $\mathbb{E}[\log (V_T^{\pi})] \leq\mathbb{E}[ \log (V_T^{\text{num}})]$ for all $\pi \in \Pi$.

Moreover, the expected value of the log-optimal portfolio is given due to \eqref{eq:portfolio} by
\[
\sup_{\pi \in \Pi}\mathbb{E}[\log V^{\pi}_T]=\frac{1}{2}\mathbb{E}\left[\int_0^T \lambda^{\top}(\mu_t) c(\mu_t) \lambda(\mu_t) dt\right].
\]

So far we have optimized over \emph{all} strategies in $\Pi$. 
In the sequel we shall mainly consider suprema taken over smaller sets, in particular over $\mathcal{FG}^{M,\alpha}$. Note that in this case the optimizer will 
still be a function of the market weights due to the time-homogeneous Markov property of $(\mu_t)_{t \geq 0}$.

In this context let us also answer the question when the log-optimal num\'eraire portfolio is functionally generated in the sense of Proposition 
\ref{prop:logoptfunctionally} below. This is needed to relate its asymptotic growth rate to the one of the best retrospectively chosen portofolio and 
the universal portfolio (see Theorem \ref{cor:main} below).

\begin{proposition}\label{prop:logoptfunctionally}
Let $(\mu_t)_{t \geq 0}$ be of the form  \eqref{eq:muP}.  Then the log-optimal portfolio is generated by a differentiable function $G$, i.e.,
\begin{align*}
(\pi_t^{\text{num}})^i=\mu^i_{t}\left(\frac{D^i G(\mu_t)}{G(\mu_t)}+1-\sum_{j=1}^d \mu^j_{t}\frac{D^j G(\mu_t)}{G(\mu_t)}\right), \quad i=1, \ldots,d,
\end{align*}
 if the drift characteristic $\lambda$ satisfies
\[
\lambda(x)=\nabla \log G(x)= \frac{\nabla G(x)}{G(x)}, \quad x \in \Delta^d.
\]
\end{proposition}

\begin{proof}
The assertion follows from expression~\eqref{eq:logopt}. 
\end{proof}

\subsubsection{Asymptotic growth rates in the case of an ergodic market weights process}\label{sec:assgrowth}

This paragraph is dedicated to establish the form of the asymptotic growth rates, 
 given by 
\[
\lim_{T\to \infty}\frac{1}{T} \log V^{\pi}_T,
\] 
when the market weights process $(\mu_t)_{t \geq 0}$ described by the stochastic model in~\eqref{eq:muP} satisfies the following ergodicity assumption.

\begin{assumption}\label{ass:1}
The process $\mu$ as given in~\eqref{eq:muP} is an ergodic process with stationary  measure $\rho$ on $\Delta^d$.
\end{assumption}

For the precise notion of ergodicity in continuous time we refer to \cite[Section 2.2., Theorem 2.4 and Section 2.2.3]{E:15}. 
In the following theorem we consider portfolio maps which are not necessarily long-only, but can take values in the hyperplane $H^d$.

\begin{theorem}\label{th:ergodic}
Under Assumption~\ref{ass:1} the following statements hold true:
\begin{enumerate}
\item Let $\pi: \Delta^d \to H^d$  be any ($\rho$-measurable) portfolio map such that
\begin{align}
\int_{\Delta^d} \left|\left(\frac{\pi(x)}{x}\right)^{\top}c(x)\lambda(x)\right|\rho(dx) < \infty,\notag\\
Q^{\pi}:=\int_{\Delta^d}\left(\frac{\pi(x)}{x}\right)^{\top} c(x)\left(\frac{\pi(x)}{x}\right)\rho(dx) < \infty.\label{eq:finitevar} 
\end{align}
We then have, for $\rho$-a.e.~starting value $\mu_0$, that
\begin{align*}
\lim_{T \to \infty} \frac{1}{T} \log(V^{\pi}_T)&=L^{\pi}:=\int_{\Delta^d} \left(\frac{\pi(x)}{x}\right)^{\top}c(x)\lambda(x)\rho(dx)\\
& -\frac{1}{2}
\int_{\Delta^d}\left(\frac{\pi(x)}{x}\right)^{\top} c(x)\left(\frac{\pi(x)}{x}\right)\rho(dx) 
, \quad \mathbb{P}\textrm{-a.s.} 
\end{align*}
\item Assume that $L^{\text{num}}:=\frac{1}{2}\int_{\Delta^d} \lambda^{\top}(x) c(x) \lambda(x)\rho(dx)< \infty$. Then, for $\rho$-a.e.~starting value $\mu_0$, it holds that
\[
\lim_{T \to \infty} \frac{1}{T}\log V^{\text{num}}_T=L^{\text{num}}, \quad \mathbb{P}\textrm{-a.s.}
\]
\end{enumerate}
\end{theorem}

The proof of Theorem \ref{th:ergodic} relies on the following lemma which is stated and proved in~\cite[Lemma 1.3.2]{F:02}.

\begin{lemma}\label{lem:conv}
Let $M$ be a continuous local martingale such that
\begin{align}\label{eq:condconv} 
\lim_{T \to \infty}\frac{1}{T^2} \langle M, M\rangle_T \log \log T=0, \quad \mathbb{P}\textrm{-a.s.}
\end{align}
Then $\lim_{T \to \infty} \frac{1}{T}M_T=0$, $\mathbb{P}$-a.s.
\end{lemma}

\begin{proof}[Proof of Theorem \ref{th:ergodic}]
Let us start by proving statement (i). By~\eqref{eq:wealth}, $\log V^{\pi}_T$ reads as
\begin{equation}
\begin{split}\label{eq:logV}
\log V^{\pi}_T&= \int_0^{T} \left(\frac{\pi(\mu_t)}{\mu_t}\right)^{\top}c(\mu_t)\lambda(\mu_t)dt-\frac{1}{2}\int_0^{T}\left(\frac{\pi(\mu_t)}{\mu_{t}}\right)^{\top}c(\mu_t)\frac{\pi(\mu_t)}{\mu_{t}}dt\\
&\quad+\int_0^{T} \left(\frac{\pi(\mu_t)}{\mu_t}\right)^{\top}\sqrt{c(\mu_t)}dW_t.
\end{split}
\end{equation}
The local martingale part
\[
M^{\pi}_T:= \int_0^T\left(\frac{\pi(\mu_t)}{\mu_t}\right)^{\top}\sqrt{c(\mu_t)} dW_t
\]
satisfies Condition~\eqref{eq:condconv} of Lemma~\ref{lem:conv} below. Indeed, by the ergodic theorem in continuous time (see e.g.,\cite[Theorem 2.4 and Section 2.2.3]{E:15}) and~\eqref{eq:finitevar} we have
\[
\frac{1}{T}\langle M^{\pi},M^{\pi}\rangle_T=
\frac{1}{T}\int_0^{T}\left(\frac{\pi(\mu_t)}{\mu_{t}}\right)^{\top}c(\mu_t)\frac{\pi(\mu_t)}{\mu_{t}}dt\stackrel{T\to \infty} \to   Q^{\pi}< \infty, \quad \mathbb{P}\textrm{-a.s.}
\]
Multiplying the left hand side with $(\log \log T)/T$, therefore yields Condition~\eqref{eq:condconv} and 
\[
\frac{1}{T} M^{\pi}_T=\frac{1}{T} \int_0^T \left(\frac{\pi(\mu_t)}{\mu_t}\right)^{\top} \sqrt{c(\mu_t)} dW_t \to 0, \quad \mathbb{P}\textrm{-a.s.}
\]
Hence, evoking again the ergodic theorem yields
\begin{align*}
\lim_{T \to \infty}\frac{1}{T}\log V^{\pi}_T&= \lim_{T \to \infty}\frac{1}{T}\left(\int_0^{T} \left(\frac{\pi(\mu_t)}{\mu_t}\right)^{\top}c(\mu_t)
\lambda(\mu_t)dt-\frac{1}{2}\int_0^{T}\left(\frac{\pi(\mu_t)}{\mu_{t}}\right)^{\top}c(\mu_t)\frac{\pi(\mu_t)}{\mu_{t}}dt\right)\\
&=\int_{\Delta^d} \left(\frac{\pi(x)}{x}\right)^{\top}c(x)\lambda(x)\rho(dx)-\frac{1}{2}
\int_{\Delta^d}\left(\frac{\pi(x)}{x}\right)^{\top} c(x)\left(\frac{\pi(x)}{x}\right)\rho(dx),
\end{align*}
$\mathbb{P}$-a.s. (and also in $L^1(\Omega,\mathcal{F},P)$) and thus assertion (i).

Concerning statement (ii), note that according to relation~\eqref{eq:scaledweights}, the scaled relative weights corresponding to the log-optimal portfolio satisfy
\[
c(x)\left(\frac{\pi^{\text{num}}(x)}{{x}}-\lambda(x)\right)=0.
\]
Thus, by~\eqref{eq:logV} and \eqref{eq:wealth}, $\log V^{\text{num}}_T$ simplifies to
\[
\log V^{\text{num}}_T= \frac{1}{2} \int_0^T \lambda^{\top}(\mu_t)c(\mu_t) \lambda(\mu_t)dt+ \int_0^T\lambda^{\top}(\mu_t)\sqrt{c(\mu_t)}dW_t.
\]
In this case we have
\[
\frac{1}{T}\langle M^{\pi^{\text{num}}},M^{\pi^{\text{num}}}\rangle_T=
\frac{1}{T}\int_0^T  \lambda^{\top}(\mu_t)c(\mu_t) \lambda(\mu_t) dt\stackrel{T\to \infty} \to  2 L^{\text{num}}, \quad \mathbb{P}\textrm{-a.s.},
\]
which yields by the same argument as above
\[
\frac{1}{T} M^{\pi^{\text{num}}}_T=\frac{1}{T} \int_0^T \lambda^{\top}(\mu_t)\sqrt{c(\mu_t)}dW_t \to 0, \quad \mathbb{P}\textrm{-a.s.}
\]
and in turn
\[
\lim_{T \to \infty}\frac{1}{T}\log V^{\text{num}}_T= \lim_{T \to \infty}\frac{1}{2T}\int_0^T \lambda^{\top}(\mu_t)c(\mu_t) \lambda(\mu_t) dt=L^{\text{num}}, \quad \mathbb{P}\textrm{-a.s.}
\]
\end{proof}

\subsection{Asymptotically equivalent growth rates}
As in discrete time we are now able to establish asymptotic equality between the growth rates of all three portfolio types introduced in Section \ref{sec:typesfun}. 
We start by comparing the best retrospectively chosen one with the universal one.

\begin{theorem} \label{th:continouscase1}
Let $M,\alpha > 0$ be fixed and let $(\mu_t)_{t \geq 0}$ be a continuous path satisfying Assumption \ref{ass:quadratic} such that for all $i \in \{1, \ldots, d\}$
\begin{align}\label{eq:assfinite}
\lim_{T \to \infty}\frac{1}{T}[\mu^i, \mu^i]_T < \infty.
\end{align} 
Consider a probability measure $m$ on $\mathcal{G}^{M,\alpha}$ with full support and set $\nu=F_* m$ with $F$ defined in \eqref{eq:Pi}. Then
\[
\lim_{T\to \infty} \frac{1}{T} (\log V_T^{*,M,\alpha} -\log V^{M,\alpha}_T(\nu)) =0,
\]
where $V^{*,M,\alpha}$ and  $V^{M,\alpha}(\nu)$ are defined in \eqref{eq:retrofunctional} and \eqref{eq:universalwealthfunk} respectively.
\end{theorem}

\begin{proof}
The inequality ``$\geq$'' is obvious. For the converse inequality we proceed similarly as in the previous section, however on the level of generating functions. 
As $m$ has full support and ${\mathcal{G}}^{M, \alpha}$ is compact, we have that, for $\eta >0$ there exists some $\delta>0$, such that every $\eta$-neighborhood of a point $G \in \mathcal{G}^{M,\alpha}$ has $m$-measure bigger than $\delta.$

Let $T \geq 1 $ and denote by $G_T^{*,M,\alpha}$ the optimizer as of Proposition \ref{prop:optimizer}.  
Consider now a generating function $G$ such $\| G- G_T^{*,M,\alpha}\|_{C^{2,0}} \leq \eta$. 
Then it follows from \eqref{eq:funcont1} that
\begin{equation}\label{eq:estimate1}
\begin{split}
\frac{1}{T}\left(\log (V_T^{G})-\log (V_T^{G_T^{*,M,\alpha}})\right)&\geq \frac{1}{T} \left(-2M \eta-\left(\frac{M}{2}d^2 \eta+
\frac{M^3}{2} d^2\eta\right)\max_{i}[\mu^i, \mu^i]_T\right)\\
&=:-K_T.
\end{split}
\end{equation}
Fix $\epsilon >0 $ and note that by assumption \eqref{eq:assfinite} and continuity of $ T\mapsto \frac{1}{T} [u^i, u^i]_T$ on $[1, \infty)$, $ \sup_{T \in [1, \infty)}\frac{1}{T}[\mu^i, \mu^i]_T$ can be bounded by some constant.  Therefore we can choose $\eta$ sufficiently small such that $K_T \leq \epsilon$ for all $T \geq 1$.

Denote by $B=B_\eta (G_T^{*,M,\alpha})$ the $\|\cdot \|_{C^{2,0}}$-ball with radius $\eta$ in $\mathcal{G}^{M,\alpha}$ which has $m$-measure at least $\delta >0,$ where $\delta$ only depends on $\eta$.
We then may estimate using Jensen's inequality and~\eqref{eq:estimate1}
\begin{align}
\left(\frac{V_T^{M,\alpha}(\nu)}{V_T^{G_T^{*,M,\alpha}}}\right)^{\frac{1}{T}}&=\left(\frac{\int_{\mathcal{G}^{M,\alpha}} V_T^{G}m(dG)}{V_T^{G_T^{*,M,\alpha}}}\right)^{\frac{1}{T}}
\geq\left(\frac{\int_{B_\eta (G_T^{*,M,\alpha})}  V_T^{G}m(dG)}{V_T^{G_T^{*,M,\alpha}}}\right)^{\frac{1}{T}}\notag\\
&\geq \delta^{\frac{1}{T}-1}
\frac{\int_{B_\eta (G_T^{*,M,\alpha})}  (V_T^{G})^{\frac{1}{T}}m(dG)}{(V_T^{G_T^{*,M,\alpha}})^{\frac{1}{T}}}\notag\\
&\geq \delta^{\frac{1}{T}}e^{-K_T} \notag\\
&\geq  \delta^{\frac{1}{T}}e^{-\epsilon}. \notag
\end{align}
Letting $T \to \infty$ for any given $\epsilon$ (which determines $\eta$ and thus in turn $\delta$) yields the assertion.
\end{proof}

In order to compare the asymptotic performance with the log-optimal portfolio we start by optimizing over portfolio maps
lying in $\mathcal{FG}^{M,\alpha}$ and suppose henceforth that $(\mu_t)_{t \geq 0}$ is a stochastic process of the form \eqref{eq:muP}.  
Under the ergodicity assumption \ref{ass:1} and 
in view of Theorem~\ref{th:ergodic} define
\begin{equation} \label{eq:logfunc}
 \begin{split}
\hat{\pi}^{M,\alpha}&:=\argmax_{\pi^G \in \mathcal{FG}^{M,\alpha}} \left(\int_{\Delta^d} \left(\frac{\pi^G(x)}{x}\right)^{\top} c(x)\lambda(x)\rho(dx)\right. \\
&\quad\left.-\frac{1}{2}
\int_{\Delta^d}\left(\frac{\pi^G(x)}{x}\right)^{\top} c(x)\left(\frac{\pi^G(x)}{x}\right)\rho(dx)\right) 
  \end{split}
\end{equation}
and the corresponding wealth process $\hat{V}^{M,\alpha}$ by $\hat{V}^{M,\alpha}=V^{\hat{\pi}^{M,\alpha}}$, whenever $\hat{\pi}^{M,\alpha}$ is well defined. Since
\[
\sup_{\pi^G \in \mathcal{FG}^{M,\alpha}} \mathbb{E}\left[ \log(V^{\pi^G}_T )\right]
\]
yields $\hat{\pi}^{M,\alpha}$ as optimizer for all $T >0$, $\hat{V}^{M,\alpha}$ corresponds to the log-optimal portfolio among functionally generated portfolios with generating function in $\mathcal{G}^{M,\alpha}$.

\begin{theorem}\label{tC3cf} 
Let $M,\alpha > 0$ be fixed and let  $(\mu_t)_{t \geq 0}$ be a stochastic process of the form \eqref{eq:muP} satisfying Assumption \ref{ass:1}. Moreover, suppose that 
\begin{align}\label{eq:conderg}
&\int_{\Delta^d} c^{ii}(x) \rho(dx)< \infty, \quad \text{for all  }i \in \{1, \ldots,d\},\\
&\int_{\Delta^d} \max_{i \in \{1, \ldots, d\}}|(c(x) \lambda(x))^i|\rho(dx) < \infty. \label{eq:conderg1}
\end{align}
Consider a probability measure $m$ on $\mathcal{G}^{M,\alpha}$ with full support and set $\nu=F_* m$ with $F$ defined in \eqref{eq:Pi}.
 Then
\begin{align}\label{eq:result1}
\liminf_{T\to \infty} \frac{1}{T} \log V_T^{*,M,\alpha} = \liminf_{T\to \infty} \frac{1}{T}\log V^{M,\alpha}_T(\nu)= \lim_{T\to \infty} \frac{1}{T} \log \hat{V}^{M,\alpha}_T, \quad \mathbb{P}\textrm{-a.s.}
\end{align}
where $\hat{V}^{M,\alpha}_T$ denotes the log-optimal portfolio among $\mathcal{FG}^{M,\alpha}$-maps defined via \eqref{eq:logfunc}, $V^{*,M,\alpha}$ and  
$V^{M,\alpha}(\nu)$ are defined pathwise in \eqref{eq:retrofunctional} and \eqref{eq:universalwealthfunk} respectively.
\end{theorem}

\begin{proof}
The well-definedness of $\hat{\pi}^{M,\alpha}$ follows 
from continuity and compactness. Indeed, the map
\begin{align*}
G &\mapsto  \int_{\Delta^d} \left(\frac{\pi^G(x)}{x}\right)^{\top} c(x)\lambda(x)\rho(dx) -\frac{1}{2}
\int_{\Delta^d}\left(\frac{\pi^G(x)}{x}\right)^{\top} c(x)\left(\frac{\pi^G(x)}{x}\right)\rho(dx)\\
&=\int_{\Delta^d} \left(\frac{\nabla G(x)}{G(x)}\right)^{\top}c(x)\lambda(x)\rho(dx) -\frac{1}{2}
\int_{\Delta^d}\left(\frac{\nabla G(x)}{G(x)}\right)^{\top} c(x)\left(\frac{\nabla G(x)}{G(x)}\right)\rho(dx)\\
\end{align*}
is continuous from $(\mathcal{G}^{M,\alpha},\| \cdot \|_{2,0})$ to $\mathbb{R}$. This together with  compactness of $\mathcal{G}^{M,\alpha}$  with respect to 
$\|\cdot\|_{2,0}$ imply the well-definedness of $\hat{\pi}^{M,\alpha}$.   

Note also that \eqref{eq:conderg} and \eqref{eq:conderg1} as well as the conditions on $G$ imply the assumptions of the ergodic theorem (Theorem \ref{th:ergodic}). Hence, we have for each $\pi^G \in \mathcal{FG}^{M,\alpha}$ the $\mathbb{P}$-a.s.~limit
\[
\lim_{T \rightarrow \infty} \frac{1}{T} \log V_T^{\pi^G} = L^{\pi^G}.
\]
In particular,
\begin{align}\label{eq:asymrate}
\lim_{T \rightarrow \infty} \frac{1}{T} \log \hat{V}^{M,\alpha}_T = \sup_{\pi^G \in {\mathcal{FG}}^{M,\alpha}} L^{\pi^G}=: L^{M,\alpha}
\end{align}
holds  $\mathbb{P}$-a.s.

Due to  \eqref{eq:conderg}, we can now apply Theorem~\ref{th:continouscase1}, which implies the first equality in \eqref{eq:result1}. 
Moreover, we have by the definition of $V^{*,M,\alpha}_T$ for each fixed $T$ the inequality
\begin{align}\label{eq:vstardom}
\frac{1}{T} \log (\hat{V}^{M,\alpha}_T) \leq \frac{1}{T} \log (V^{*,M,\alpha}_T), \quad \mathbb{P}\textrm{-a.s.}
\end{align}
Using \eqref{eq:asymrate}, \eqref{eq:vstardom} and Theorem~\ref{th:continouscase1}, we thus have $\mathbb{P}$-a.s.,
\begin{align}\label{eq:pathwiselim1}
L^{M,\alpha}=\lim_{T\to \infty}\frac{1}{T} \log (\hat{V}^{M,\alpha}_T) \leq \liminf_{T \to \infty} \frac{1}{T} \log (V^{*,M,\alpha}_T)=
\liminf_{T\to \infty} \frac{1}{T} \log (V^{M,\alpha}_T(\nu) ).
\end{align}
On the other hand, by the definition of $(\hat{V}^{M,\alpha}_t)_{t \geq 0}$ as log-optimizer within the class $\mathcal{FG}^{M,\alpha}$
\begin{align}\label{eq:expectdom1}
 \mathbb{E}[\log (V^{M,\alpha}_T(\nu)) ]\leq \sup_{\pi^G \in \mathcal{FG}^{M,\alpha}}\mathbb{E}[\log (V^{\pi^G}_T)]=\mathbb{E}[\log (\hat{V}^{M,\alpha}_T)]
\end{align}
holds. Note here that the universal portfolio  $\pi^{M,\nu, \alpha}_T$ to build $V^{M,\alpha}_T(\nu)$ is at each time a random mixture of functionally generated portfolios 
lying in $\mathcal{G}^{M,\alpha}$. 
By the time-homogenous Markovianity it is thus sufficient to dominate the left hand side of \eqref{eq:expectdom1} by taking the supremum over elements in 
$\mathcal{FG}^{M,\alpha}$.

Combining now \eqref{eq:expectdom1}, 
Theorem \ref{th:ergodic} and \eqref{eq:pathwiselim1} yields,
\begin{align*}
 \mathbb{E}[\liminf_{T\to \infty}\frac{1}{T} \log (V^{M,\alpha}_T(\nu))]&\leq \liminf_{T\to \infty} \frac{1}{T} \mathbb{E}[\log (V^{M,\alpha}_T(\nu)) ]\\
 &\leq \lim_{T\to \infty} \frac{1}{T} \mathbb{E}[\log (\hat{V}^{M,\alpha}_T)]\\
 &= \lim_{T\to \infty}\frac{1}{T} \log (\hat{V}^{M,\alpha}_T) \\
 &\leq \liminf_{T\to \infty} \frac{1}{T} \log (V^{M,\alpha}_T(\nu)), \quad \mathbb{P}\text{-a.s.},
\end{align*}
where the first inequality follows from Fatou's lemma. From this we see that 
\[
\liminf_{T\to \infty} \frac{1}{T} \log (V^{M,\alpha}_T(\nu))
\]
is $\mathbb{P}$-almost surely constant and equal to 
$\lim_{T\to \infty}\frac{1}{T} \log (\hat{V}^{M,\alpha}_T) $. Hence the assertion is proved.
\end{proof}

As in the previous section we can formulate a result not depending explicitly on the constant $M$ on $\alpha$. Setting $\alpha=\frac{1}{M}$ we choose for $M=1,2,3, \dots$ a measure $m^M$ on $\mathcal{G}^{M,\frac{1}{M}}$ with full support. Define $m=\sum^{\infty}_{M=1} 2^{-M} m^M$ and the process $V(\nu)$  by
\begin{align*}
V_T(\nu)=\int_{\bigcup^\infty_{M=1}\mathcal{G}^{M,\frac{1}{M}}} V^G_T m(dG).
\end{align*}
In order to compare the performance with the one of the global log-optimal portfolio, whenever it is functionally generated, 
we combine the above results with Proposition \ref{prop:logoptfunctionally}.

\begin{corollary}\label{cor:main}
Let $(\mu_t)_{t \geq 0}$ be a stochastic process of form \eqref{eq:muP} satisfying Assumption \ref{ass:1}. 
Moreover, suppose that its characteristics $\lambda$ and $c$ satisfy \eqref{eq:conderg} and
\begin{align}\label{ass:lambda1}
&\lambda(x)= \frac{\nabla \hat{G}(x)}{\hat{G}(x)},\\
&\frac{1}{2}\int_{\Delta^d} \frac{\nabla \hat{G}(x)}{\hat{G}(x)} c(x) \frac{\nabla \hat{G}(x)}{\hat{G}(x)} \rho(dx) < \infty \label{ass:welldef}
\end{align}
for some concave function $\hat{G} \in C^2(\bar{\Delta}^d)$. 
Then we have $\mathbb{P}$-a.s.
\begin{align}
\lim_{M\to \infty} \lim_{T \to \infty} \frac{1}{T} \log ( V_{T}^{*,M,\frac{1}{M}} ) = \lim_{T \to \infty} \frac{1}{T} \log (V_{T}(\nu) )=
\lim_{{T}\to \infty} \frac{1}{T} \log (V^{\text{num}}_{T})=L^{\text{num}},
\end{align}
where $L^{\text{num}}$ is given by \eqref{ass:welldef}.
\end{corollary}

\begin{proof}
Note first that $L^{\text{num}}$ is well defined due to \eqref{ass:welldef}. 
Furthermore, note that for every $\varepsilon >0$, there exists some $M >0$ and some function $G \in \mathcal{G}^{M,\frac{1}{M}}$ such that
\begin{align*}
\lim_{T\to \infty} \frac{1}{T}\log(V^{G}_T)\geq \lim_{T\to \infty} \frac{1}{T} \log (V^{\text{num}}_T)+\varepsilon.
\end{align*}
Indeed this simply follows from continuity of $G \mapsto V^G$ as asserted in Lemma 
 \ref{lem:continuity} and  by choosing  $G \in \mathcal{G}^{M,\frac{1}{M}}$ close enough with respect to the $\|\cdot \|_{C^{2,0}}$ to the optimizing function $\hat{G} \in C^2(\bar{\Delta}^d)$ whose generated portfolio yields $V^{\text{num}}$ due to \eqref{ass:lambda1} and Proposition \ref{prop:logoptfunctionally}.
By Theorem~\ref{tC3cf}, we can therefore conclude that 
\begin{align}
\lim_{M\to \infty} \liminf_{T \to \infty} \frac{1}{T} \log ( V_{T}^{*,M,\frac{1}{M}} ) = \liminf_{T \to \infty} \frac{1}{T} \log (V_{T}(\nu) )=
\lim_{{T}\to \infty} \frac{1}{T} \log (V^{\text{num}}_{T})=L^{\text{num}}
\end{align}
holds true. Since Theorem \ref{th:continouscase1} implies that
\[
\limsup_{T\to \infty} \frac{1}{T} \log V_T(\nu) = \lim_{M \rightarrow \infty} \limsup_{T \rightarrow \infty} \frac{1}{T} \log  V_T^{*, M, \frac{1}{M}},
\]
the assertion is proved if 
\begin{align}\label{eq:finite1}
\limsup_{T\to \infty} \frac{1}{T} (\log V_T(\nu)-\log V^{\text{num}}_T) = \limsup_{T\to \infty} \frac{1}{T} \log \left(\frac{V_T(\nu)} {V^{\text{num}}_T}\right) =0, \quad \mathbb{P}\text{-a.s.}
\end{align}
By the considerations of Section \ref{sec:logoptimalfunc} (see also \cite[Propostion 4.3]{becherer2001numeraire}), it follows that $(\frac{V_t(\nu)}{V^{\text{num}}_t})_{t \geq 0}$  is a non-negative supermartingale. It converges $\mathbb{P}$-a.s.~to a finite limit as $t \to \infty$. This in turn implies \eqref{eq:finite1} and proves the statement.
\end{proof}

\vskip10pt

\appendix
\section{Proof of Theorem \ref{t1.1}}

\begin{proof}
Fix $T>0$ and the trajectory $s=(s^1_t, \dots, s^d_t)^T_{t=0} \in (\mathbb{R}^d)^{T+1}.$ For fixed $s$ the function $b \mapsto V_T(b)(s)$ is continuous on $\bar{\Delta}^d.$ Hence there is $\bar{b}=\bar{b}(s) \in \bar{\Delta}^d$ such that
\begin{align}\label{p1}
V^*_T(s)=V_T(\bar{b})(s).
\end{align}
In fact, condition \eqref{p3} implies that the sequence of functions $(b \mapsto \frac{1}{T} \log V_T(b))^\infty_{T=1}$ is Lipschitz on $\bar{\Delta}^d,$ uniformly in $T \in \mathbb{N}$ and $s$ satisfying \eqref{p3} for some fixed constants $C> c > 0.$

Indeed, consider the distance on $\bar{\Delta}^d$ defined by $\|b - \tilde{b} \|_1= \sum^d_{j=1}|b^j - \tilde{b}^j|.$ Then we may estimate
\begin{align*}
|\frac{1}{T} \log V_T(b) - \frac{1}{T} \log V_T(\tilde{b})| \leq (\log (C)- \log(c) )\|b-\tilde{b}\|_1.
\end{align*}
For $\epsilon >0$ we may therefore define $\delta := \frac{c\epsilon}{C} >0$ such that, for every $\delta$-neighbourhood $U(\bar{b})$ around any $\bar{b} \in \bar{\Delta}^d$ we have
\begin{align*}
\frac{1}{T} \log V_T(b) \geq \frac{1}{T} \log V_T(\bar{b}) - \epsilon, 
\end{align*}
for every $b \in U(\bar{b}).$ If the probability measure $\nu$ has full support, we also may find $\eta=\eta(\epsilon,c,C)>0$ such that each such $\delta$-neighborhood $U(\bar{b})$, where $\bar{b}$ runs through $\bar{\Delta}^d,$ satisfies $\nu(U(\bar{b}))>\eta.$ Using \eqref{p1} we therefore may conclude, similarly as in \eqref{C4}, that \eqref{p3a} holds true, uniformly in $s=(s^1_t, \dots, s^d_t)^\infty_{t=0}$ satisfying \eqref{p3} for some fixed constants $C > c > 0.$
\end{proof}

%\bibliographystyle{abbrv}

%\bibliography{referencesCover}

\end{document}